\documentclass[12pt,draftclsnofoot,journal,onecolumn]{IEEEtran}
\IEEEoverridecommandlockouts

\usepackage[english]{babel}
\usepackage[utf8x]{inputenc}
\usepackage{amsmath,amssymb}
\usepackage{amsthm}
\usepackage[ruled,linesnumbered]{algorithm2e}
\usepackage{graphicx}
\usepackage{textcomp}
\usepackage{xcolor}
\usepackage{stfloats}
\usepackage{multirow}
\usepackage{verbatim}
\usepackage{textcomp}
\usepackage{caption}
\usepackage{subfigure}
\usepackage{mathtools}
\usepackage{cite}
\usepackage{tikz,hyperref}
\usepackage{comment}

\def\BibTeX{{\rm B\kern-.05em{\sc i\kern-.025em b}\kern-.08em
    T\kern-.1667em\lower.7ex\hbox{E}\kern-.125emX}}
    
\newtheorem{theorem}{Theorem}

\DeclareMathAlphabet\mathcal{OMS}{cmsy}{m}{n}
\SetMathAlphabet\mathcal{bold}{OMS}{cmsy}{b}{n}

\DeclarePairedDelimiter\ceil{\lceil}{\rceil}

\definecolor{mod}{rgb}{0,0,0} 

\newcommand\figSize{0.475}

\newcommand{\orcidicon}{\includegraphics[width=0.32cm]{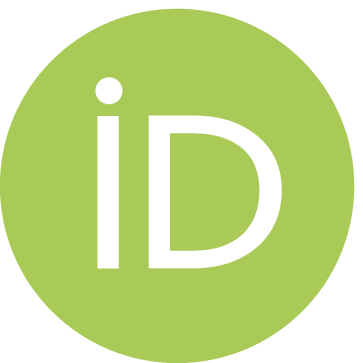}}

\foreach \x in {A, ..., Z}{%
	\expandafter\xdef\csname orcid\x\endcsname{\noexpand\href{https://orcid.org/\csname orcidauthor\x\endcsname}{\noexpand\orcidicon}}
}

\allowdisplaybreaks

\begin{document}

\title{Learning from Images: Proactive Caching with Parallel Convolutional Neural Networks}

\author{Yantong~Wang,~Ye~Hu,~Zhaohui~Yang,~Walid~Saad,~Kai-Kit~Wong,~and~Vasilis~Friderikos
	\thanks{
	\IEEEcompsocthanksitem Y. Wang and V. Friderikos are with the Center of Telecommunication Research, Department of Engineering, King's College London, London, WC2R 2LS, UK. \protect E-mail: \{yantong.wang,vasilis.friderikos\}@kcl.ac.uk.
	\IEEEcompsocthanksitem Y. Hu and W. Saad are with the Wireless@VT, Bradley Department of Electrical and Computer Engineering, Virginia Tech, Balcksburg, VA 24061 USA. \protect E-mail: \{yeh17,walids\}@vt.edu.
	\IEEEcompsocthanksitem Z. Yang and K. Wong are with the Department of Electronic and Electrical Engineering, University College London, London, WC1E 6BT, UK. \protect E-mail: \{zhaohui.yang,kai-kit.wong\}@ucl.ac.uk.
	}\vspace{-5em}
} 

\maketitle
\IEEEpeerreviewmaketitle

\begin{abstract}
With the continuous trend of data explosion, delivering packets from data servers to end users causes increased stress on both the fronthaul and backhaul traffic of mobile networks. To mitigate this problem, caching popular content closer to the end-users has emerged as an effective method for reducing network congestion and improving user experience. To find the optimal locations for content caching, many conventional approaches construct various mixed integer linear programming (MILP) models. However, such methods may fail to support online decision making due to the inherent curse of dimensionality. In this paper, a novel framework for proactive caching is proposed. This framework merges model-based optimization with data-driven techniques by transforming an optimization problem into a grayscale image. For parallel training and simple design purposes, the proposed MILP model is first decomposed into a number of sub-problems and, then, convolutional neural networks (CNNs) are trained to predict content caching locations of these sub-problems. Furthermore, since the MILP model decomposition neglects the internal effects among sub-problems, the CNNs' outputs have the risk to be infeasible solutions. Therefore, two algorithms are provided: the first uses predictions from CNNs as an extra constraint to reduce the number of decision variables; the second employs CNNs' outputs to accelerate local search. Numerical results show that the proposed scheme can reduce $71.6\%$ computation time with only $0.8\%$ additional performance cost compared to the MILP solution, which provides high quality decision making in real-time. 
\end{abstract}

\section{Introduction}
\label{sec:introduction}
Caching popular contents at the edge of a wireless network has emerged as an effective technique to alleviate congestion and data traffic over the backhaul and fronthaul of existing wireless systems~\cite{SurveyCaching}. Caching methods can be classified into three categories~\cite{vasilakos2012proactive}: reactive caching~\cite{sourlas2010mobility}, durable subscription~\cite{farooq2004performance}, and proactive caching~\cite{gaddah2010extending}. In reactive caching, the connected access point keeps caching items that match user's subscription after the user disconnects from the network. When the user reconnects from a different point later in time, the caching items are retrieved from the previous access point. Reactive caching has a re-transmission delay from the old access point to the new one. In durable subscription, both the current connected access point and all proxies within a domain that a user can possibly connect to should keep caching the content. Therefore, durable subscription reduces the transmission delay but significantly increases memory usage. Compared with those two policies, proactive caching selects a subset of proxies as potential caching hosts thus giving rise to a balance between storage cost and transmission latency. However, enabling proactive caching in real networks faces four key challenges~\cite{wang2020survey}: where to cache, what to cache, cache dimensioning, and content delivery. In this paper, the primary concentration is on ``where to cache'' problem, i.e., determining caching locations.

Proactive caching techniques have attracted significant attention in prior art~\cite{yang2018cache,yang2020joint,wang2019proactive,zheng2016optimal,fang2015energy,zou2019joint}. Conventionally, the ``where to cache'' problem is modeled as an optimization problem that is then solved using convex optimization~\cite{yang2018cache,yang2020joint}, mixed integer linear programming (MILP)~\cite{wang2019proactive,zheng2016optimal}, and game theory~\cite{fang2015energy,zou2019joint}. However, these solutions, particularly MILP, can be $NP$-hard. Due to the curse of dimensionality, such optimal model-based methods are not suitable to support real-time decision making on selecting caching locations.

Recently, deep learning (DL) has emerged as an important tool for solving caching problems as discussed in~\cite{lei2019learning,lei2017deep,chen2017caching,tsai2018caching,ale2019online,fan2021pa,qian2020reinforcement,he2017integrated,li2019deep,zhong2020deep}. In~\cite{lei2019learning}, the authors propose the use of convolutional neural network (CNN) to perform time slot allocation for content delivery at wireless base stations. The study in~\cite{lei2017deep} employs a fully-connected neural network (FNN) to simplify the searching space of caching optimization model and, then, determine caching base stations. In~\cite{chen2017caching}, the authors introduced a conceptor-based echo state network (ESN) to learn the users' split patterns independently, which results in a more accurate prediction. The authors in~\cite{tsai2018caching} propose a long short-term memory (LSTM) model to analyse and extract information from 2016 U.S. election tweets, thus providing support for content caching. The authors in~\cite{ale2019online} propose a stacked DL structure, which consists of a CNN, a bidirectional LSTM and a FNN, for online content popularity prediction and cache placement in the mobile network. In~\cite{fan2021pa}, a gated recurrent unit (GRU) model is proposed to predict time-variant video popularity for cache eviction. Except for the aforementioned supervised learning methods, a number of works used deep reinforcement learning (DRL) for content caching. For instance, the authors in~\cite{qian2020reinforcement} propose a joint caching and push policy for mobile edge computing networks using a deep Q network (DQN). In~\cite{he2017integrated}, the authors employ a double dueling DQN for dynamic orchestration of networking, caching and computing in vehicle networks. Moreover, some actor-critic models are proposed in~\cite{li2019deep} and~\cite{zhong2020deep} to decide content caching in wireless networks.

However, most existing DL works on caching such as~\cite{lei2019learning,lei2017deep,chen2017caching,tsai2018caching,ale2019online} and~\cite{qian2020reinforcement,he2017integrated,li2019deep,zhong2020deep} have considered a two-tier heterogeneous network structure, in which the caching content is available within one or two hops transmission. Therefore, these works may not scale well in flow routing decisions of multi-hop environments. The study in~\cite{fan2021pa} has focused on the caching policy of a single node, which may results in insufficient utilization of the network caching resources. Additionally, although the well-trained DL models can provide competitive solutions compared with conventional heuristics, the training process is very time-consuming and struggle to converge, particularly LSTM in~\cite{tsai2018caching,ale2019online} and DRL in~\cite{qian2020reinforcement,he2017integrated,li2019deep,zhong2020deep}. 

In this paper, the main contribution is to propose a framework merging MILP model with data-driven approach to determine caching locations. The proposed MILP model jointly considers the caching cost and multi-hop transmission cost in mobile network, with constraints of node storage capacity and link bandwidth limitation. Since the MILP model shares some properties with image recognition task (more details are represented in section~\ref{subsubsec:CNN}), we provide a novel approach to transform the MILP model into a grayscale image. To our best knowledge, beyond our previous work in~\cite{wang2020caching}, no work has used this method for image transformation. Recently, CNN has been widely used in computer vision tasks and the study in~\cite{sze2017efficient} shows that a CNN can exceed human-level accuracy in image recognition. Therefore, we consider CNN as the data-driven approach to extract spatial features from the input images. In order to accelerate CNN training, we decompose the MILP model into a number of independent sub-problems, and, then train CNNs in parallel with optimal solutions. As aforementioned, the caching assignment problem is $NP$-hard, and, thus, calculating optimal solutions can be time-consuming, especially when dealing with large search space instances. Note that for $NP$-hard problems, many small to medium search space instances can be solved efficiently~\cite{nowak2018revised}. Therefore, we solve the small instances of caching assignment problem to train CNNs, and, then use trained CNNs to predict caching locations for medium and large instances. To further improve the performance, two algorithms are provided: a reduced MILP model using CNNs' outputs as an extra constraint; and a hill-climbing local search algorithm using CNNs' output as searching directions. Unlike the work in~\cite{lei2019learning} and~\cite{lei2017deep}, which require the number of requests should exactly match the input layer of trained CNN, we propose an augmenting allocations algorithm to deal with the requests excessive case. We show via numerical results that the proposed framework can reduce $71.6\%$ computation time with only $0.8\%$ additional performance cost compared with the optimal solution. Furthermore, the impact of hyperparameters on CNN's performance is also presented. To this end, the key contributions of this paper can be summarized as follows,
\begin{itemize}
    \item We propose a framework which merges MILP model with CNNs to solve ``where to cache'' problem. We also provide a novel method that transforms the MILP model into a grayscale image to train the CNNs.
    \item For simple design and parallel training purposes, we decompose 
    the MILP model and train CNNs correspondingly to predict content caching locations. 
    \item To further improve the performance, we propose two algorithms which perform global and local search based on CNNs' prediction. 
    \item To deal with requests excessive case, we give an augmenting allocations algorithm to divide the input image into partitions that match CNN's input size.
\end{itemize}

The rest of paper is organized as follows. Section \ref{sec:model} presents the system model for proactive caching. In Section \ref{sec:DNN}, we discuss the transformation of the MILP model into a grayscale image as well as the training and testing process of the proposed CNN framework. Section \ref{sec:investigations} then provides a performance evaluation and Section \ref{sec:conclusions} concludes the paper and outlines future works.
\section{System Model and Problem Formulation}
\label{sec:model}

We model a mobile network as an undirected graph $\mathcal{G}=\{\mathcal{V},\mathcal{L}\}$ having a set of vertices $\mathcal{V}$ and a set of links $\mathcal{L}$, as shown in Figure \ref{fig:system}. We assume that vertices consist of access points or access routers\footnote{Here the term \textit{router} is not the core router in the Internet backbone but the Information-Centric Networking (ICN) router~\cite{xylomenos2013survey}.} (ARs) set $\mathcal{A}$ and general routers set $\mathcal{R}$, where the former provides network connectivity and the latter supplies traffic flow forwarding. Meanwhile, some of the access and typical routers have caching capacity to host the users' contents of interest. We call these routers content routers or edge clouds (ECs)\footnote{The term \textit{edge cloud} and \textit{content router} will be used interchangeably.}. We define $\mathcal{E}\subset\mathcal{A}\cup\mathcal{R}$ as the set of ECs. In this model, the end-users are mobile and, thus, their serving ARs will change over time. Here we assume that the mobility pattern of each user is obtainable, since the user moving probability is predicable by exploring historical data~\cite{chen2017caching}.

\begin{figure}[t]
	\centering
	\includegraphics[trim=0mm 0mm 0mm 0mm, clip, width=.8\textwidth]{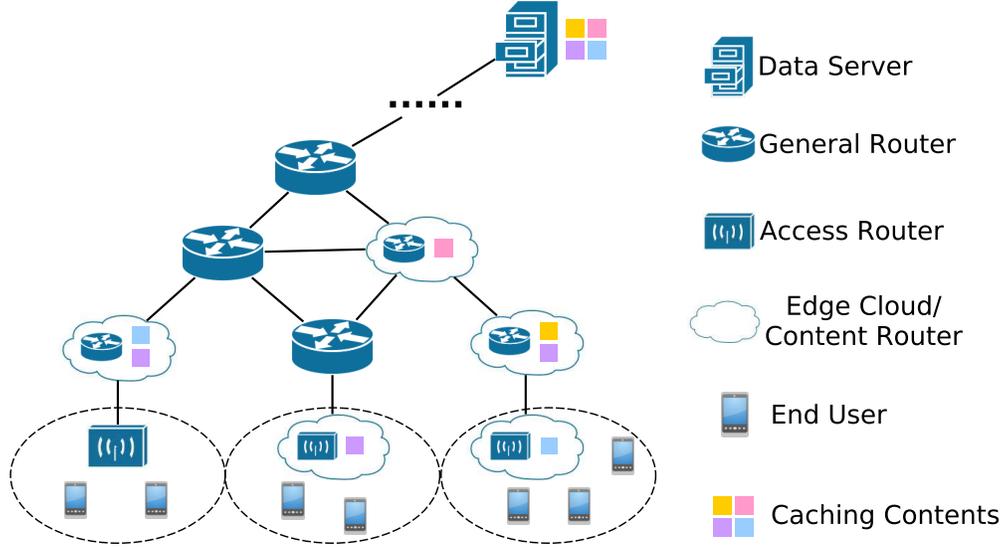}
	\caption{Illustration of our caching system model.}
	\label{fig:system}
\end{figure} 

We then define $w_e$ as the caching storage capacity of each content router $e\in\mathcal{E}$ and $c_l$ as the bandwidth of each communication link $l\in\mathcal{L}$. Here, without loss of generality, the salient assumption is that each user requests only one flow\footnote{The symbol $k$ represents both the end-user and the associated request flow.} $k$ and the set of traffic flows transmitted through the network is given by $\mathcal{K}$. 
Let $s_k$ be the required data size of user $k$; $b_k$ be the required transmission rate of user $k$; and $p_{ka}$ be the probability of user $k$ connecting to AR $a\in\mathcal{A}$. 
In order to measure the transmission cost, we define $N_{ae}$ as the number of hops from AR $a$ to a content hosting router $e$ where $N_{ae}$ is obtained via the shortest path (unambiguously $N_{ae}=0$ holds when $a=e$). $N^T$ is the number of hops from an AR to the central data server. According to \cite{van2014performance}, from a given source to a given destination, there are 10 to 15 hops on average for each user's request. For simplicity, $N^T$ is considered to be a constant equal to the mean number of hops from the AR to the data server. Our notation is summarized in Table \ref{tab:Notations}.

\begin{table}[t] 
	\centering 
	\caption{Summary of main notations.}
	\begin{tabular}{c|l} 
		\hline
		\hline
		$\mathcal{K}$ & set of flows/requests from users\\
		$\mathcal{L}$ & set of transmission links \\
		$\mathcal{A}$ & set of ARs \\
		$\mathcal{E}$ & set of ECs \\
		$|\mathcal{X}|$ & cardinality of set $\mathcal{X}$ \\
		\hline
		$\alpha$ & impact factor of caching cost\\
		$\beta$ & impact factor of transmission cost\\
		$N_{ae}$ & number of hops from AR $a$ to EC $e$\\
		$B_{lae}$ & relation between link $l$ and shortest path from $a$ to $e$\\
		$N^T$ & number of hops from AR to data center\\
		\hline
		$p_{ka}$ & the probability of flow $k$ connecting with AR $a$\\
		$s_k$ & flow $k$ required caching size\\
		$b_k$ & flow $k$ required bandwidth\\
		$w_e$ & available space on EC $e$\\
		$c_l$ & available link capacity on link $l$\\
		\hline
		\hline
	\end{tabular}
	\label{tab:Notations}
\end{table}

\subsection{Problem Formulation}

We consider two types of cost: a content caching cost $C^C$ and a content transmission cost $C^T$.
From~\cite{vasilakos2012proactive}, we have:
\begin{equation}
\label{fml:cc}
C^C(x_{ke})=\sum_{k\in\mathcal{K}}\sum_{e\in\mathcal{E}}\frac{x_{ke}}{1-\sum_{k\in\mathcal{K}}s_k\!\cdot\! x_{ke}/w_e},
\end{equation}
where $x_{ke}$ is a binary decision variable with $x_{ke}=1$ indicating that the content for flow $k$ is cached at EC $e$, otherwise, we have $x_{ke}=0$.
Therefore, $\sum_{k\in\mathcal{K}}s_k\!\cdot\! x_{ke}/w_e$ represents the storage utilization degree of EC $e$. From formula \eqref{fml:cc}, higher storage utilization leads to a higher cost for content caching and, thus, $C^C$ can be used for load balancing.

Next, we define the transmission cost $C^T$ that depends on the number of traversal hops $N_{ae}$ and $N^T$:
\begin{equation}
\begin{aligned}
\label{eq:trans}
C^T(z_{kae})=\sum_{k\in\mathcal{K}}\sum_{a\in\mathcal{A}}\sum_{e\in\mathcal{E}} p_{ka}z_{kae}N_{ae}+\sum_{k\in\mathcal{K}}(1\!-\!\sum_{a\in\mathcal{A}}\sum_{e\in\mathcal{E}}p_{ka}z_{kae})N^T,
\end{aligned}
\end{equation}
where $z_{kae}$ is a binary decision variable with $z_{kae}=1$ indicating that flow $k$ connects with AR $a$ and receives the requested content from EC $e$ through the shortest path, otherwise, $z_{kae}=0$.
Therefore, $p_{ka}z_{kae}$ represents the probability for flow $k$ choosing the shortest path between AR $a$ and EC $e$ and $\sum_{a\in\mathcal{A}}\sum_{e\in\mathcal{E}}p_{ka}z_{kae}$ evaluates the probability of cache-hit happens for flow $k$. Hence, $\sum_{k\in\mathcal{K}}\sum_{a\in\mathcal{A}}\sum_{e\in\mathcal{E}} p_{ka}z_{kae}N_{ae}$ is the expected number of transmission hops in the cache-hit case. Similarly, $(1\!-\!\sum_{a\in\mathcal{A}}\sum_{e\in\mathcal{E}}p_{ka}z_{kae})$ captures the probability for cache miss and $\sum_{k\in\mathcal{K}}(1\!-\!\sum_{a\in\mathcal{A}}\sum_{e\in\mathcal{E}}p_{ka}z_{kae})N^T$ is the expected number of transmission hops accordingly.

Our goal is to minimize the total cost function $J$ given by:
\begin{equation}
J(x_{ke},z_{kae})=\alpha\cdot C^C+\beta\cdot C^T,
\end{equation}
where $\alpha$ and $\beta$ are the weights for caching cost $C^C$ and transmission cost $C^T$ respectively. We can now formally pose the proactive edge caching problem as follows:
\begin{subequations} \label{LP:main}
	\begin{align}
	\label{LP:obj}
	&\mathop{\min}\; J(x_{ke},z_{kae}).\\ 
	\textrm{s.t.}\quad
	\label{LP:con1}
	& \sum_{e\in\mathcal{E}} x_{ke}\!=\!1, \forall k\!\in\!\mathcal{K}, \\
	\label{LP:con2}
	& \sum_{k\in\mathcal{K}} s_k\!\cdot\!x_{ke}\!<\!w_e, \forall e\!\in\!\mathcal{E},\\
	\label{LP:con3}
	& \sum_{k\in\mathcal{K}} b_k\!\cdot\!y_{kl}\!<\!c_l,\forall l\in\mathcal{L}, \\
	\label{LP:con4}
	& \sum_{e\in\mathcal{E}} z_{kae}\!\leq\!1, \forall k\!\in\!\mathcal{K},a\!\in\!\mathcal{A}, \\
	\label{LP:con5}
	& z_{kae}\!\leq\!x_{ke}, \forall k\!\in\!\mathcal{K},a\!\in\!\mathcal{A},e\!\in\!\mathcal{E},\\
	\label{LP:con6}
	& y_{kl}\!\leq\!\sum_{a\in\mathcal{A}}\sum_{e\in\mathcal{E}} B_{lae}\!\cdot\!z_{kae}, \forall k\!\in\!\mathcal{K},l\!\in\!\mathcal{L}, \\
	\label{LP:con7}
	& M\!\cdot\!y_{kl}\!\geq\!\sum_{a\in\mathcal{A}}\sum_{e\in\mathcal{E}} B_{lae}\!\cdot\!z_{kae}, \forall k\!\in\!\mathcal{K},l\!\in\!\mathcal{L},\\
	\label{LP:con8}
	& x_{ke},y_{kl},z_{kae}\!\in\!\{0,1\},\forall k\!\in\!\mathcal{K},l\!\in\!\mathcal{L},a\!\in\!\mathcal{A},e\!\in\!\mathcal{E}.
	\end{align}
\end{subequations}
In problem \eqref{LP:main}, $M$ is a sufficiently large number. 
Constraint \eqref{LP:con1} means that each flow will be served by exactly one EC. \eqref{LP:con2} and \eqref{LP:con3} capture the capacity of each EC storage and individual link bandwidth respectively. Here, $y_{kl}$ is a binary decision variable with $y_{kl}=1$ implying that flow $k$ traverses via link $l$; otherwise $y_{kl}=0$. \eqref{LP:con4} indicates that the redirected path is unique and \eqref{LP:con5} enforces an EC should be selected as a caching host, if any flow retrieves the caching content from this EC. 
Furthermore, the next two constraints \eqref{LP:con6} and \eqref{LP:con7} mean that the link being selected for transmission should be on the path between AR $a$ and EC $e$ as per the decision variable $z_{kae}$ (and vice versa). $B_{lae}$ describes the relationship between link and retrieved path, which is defined based on the network topology: 
$$B_{lae}=
\begin{cases}
1, &\text{link $l$ belongs to the shortest path between AR $a$ and EC $e$,} \\
0, &\text{otherwise.}
\end{cases}$$  

\subsection{Model Linearization}
The denominator in \eqref{fml:cc} contains the decision variable $x_{ke}$, and, hence, it results in a non-linear part in the objective function of \eqref{LP:main}. To linearize it, we define a new variable:
\begin{equation}
\label{def:t_e}
t_e=\frac{1}{1-\sum_{k\in\mathcal{K}}s_k\!\cdot\!x_{ke}/w_e}, \forall e\in\mathcal{E}.
\end{equation}
This definition of $t_e$ can be converted to the following constraints:
\begin{subequations}
\begin{align}
\label{con:t_e1}
&t_e-\sum_{k\in\mathcal{K}}\frac{s_k}{w_e}\!\cdot\!x_{ke}t_e=1, \forall e\in\mathcal{E},\\
\label{con:t_e2}
&t_e>0, \forall e\in\mathcal{E}.
\end{align}
\end{subequations}
Then, the caching cost $C^C$ becomes
\begin{equation}
\label{fml:cc2}
    C^C(x_{ke},t_e)=\sum_{k\in\mathcal{K}}\sum_{e\in\mathcal{E}}x_{ke}t_e.
\end{equation}
Both \eqref{con:t_e1} and \eqref{fml:cc2} contain a product of two decision variables (i.e. $x_{ke}t_e$), so we introduce an auxiliary decision variable $\chi_{ke}$, as follows:
\begin{equation}
    \chi_{ke}=x_{ke}t_e=
    \begin{cases}
    t_e,&\text{$x_{ke}$ is 1,}\\
    0,&\text{otherwise.}
    \end{cases}
\end{equation}
Clearly, $\chi_{ke}$ is the product of a continuous variable ($t_e$) and a binary variable ($x_{ke}$), and, thus, it must satisfy the following constraints:
\begin{subequations}
\begin{align}
	\label{LP:con9}
	& \chi_{ke}\!\leq\!t_e,\forall k\!\in\!\mathcal{K},e\!\in\!\mathcal{E},\\
	\label{LP:con10}
	& \chi_{ke}\!\leq\!M\!\cdot\!x_{ke},\forall k\!\in\!\mathcal{K},e\!\in\!\mathcal{E},\\
	\label{LP:con11}
	& \chi_{ke}\!\geq\!M\!\cdot\!(x_{ke}-1)\!+\!t_e,\forall k\!\in\!\mathcal{K},e\!\in\!\mathcal{E}.
\end{align}
\end{subequations}
As defined above, $M$ is a sufficiently large number. Therefore, \eqref{con:t_e1} and \eqref{fml:cc2} can be rewritten in terms of $\chi_{ke}$ as follows:
\begin{equation}
\label{con:t_e}
t_e-\sum_{k\in\mathcal{K}}\frac{s_k}{w_e}\chi_{ke}=1,\forall e\in\mathcal{E},
\end{equation}
\begin{equation}
\label{fml:cc3}
C^C(\chi_{ke})=\sum_{k\in\mathcal{K}}\sum_{e\in\mathcal{E}}\chi_e.
\end{equation}
Finally, we can write the following MILP model:
\begin{subequations} \label{LP:main_MILP}
	\begin{align}
	\label{LP:obj_MILP}
	&\mathop{\min}\; J(\chi_{ke},z_{kae})\\ 
	\textrm{s.t.}\quad
	&\eqref{LP:con1}\sim\eqref{LP:con8},\eqref{LP:con9}\sim\eqref{LP:con11},\eqref{con:t_e},\nonumber \\
	\label{LP:con12}
	& t_e,\chi_{ke}>0,\forall k\!\in\!\mathcal{K}, e\!\in\!\mathcal{E}.
	\end{align}
\end{subequations}
Regarding the time complexity of the aforementioned model, we have the following theorem:
\begin{theorem}
	\label{theo:np_hard}
	The MILP model \eqref{LP:main_MILP} falls into the family of $NP$-hard problems.
\end{theorem}
\begin{proof}
	See Appendix \ref{sec:proof_A}.
\end{proof}
Therefore, with the increment of the number of requests, it is quite time-consuming to solve \eqref{LP:main_MILP}. In order to accelerate the solving process and overcome its $NP$-hard nature, we propose a framework which merges the optimization model with deep learning method in the next section.
\section{Deep Learning for Caching}
\label{sec:DNN}

In this section, we use CNN to capture the spatial features of \eqref{LP:main_MILP}. Solving \eqref{LP:main_MILP} allows us to determine the allocation of each requested content, i.e. $x_{ke}$, which is expected to be the output of the CNN. However, each EC $e$ can serve more than one flow if there is enough memory space for caching, and in this case, the size of the CNN output space grows exponentially according to $|\mathcal{E}|\cdot 2^{|\mathcal{K}|}$. To deal with this challenge, the dependency among ECs and flows is taken into consideration. Constraint \eqref{LP:con1} forces each flow to be served by exactly one EC. If we ignore the caching resource competition among flows, i.e. the storage space in constraint \eqref{LP:con2} and the bandwidth in \eqref{LP:con3}, then, we can decompose the original caching allocation problem into a number of independent sub-problems where every sub-problem focuses on the allocation of one flow, i.e., we use a first-order strategy~\cite{zhang2013review}.
Moreover, in each sub-problem, the decision variable $x_{ke}$ is reduced to $x_{e}$, and constraint \eqref{LP:con1} becomes $\sum_{e\in\mathcal{E}}x_e=1$. In other word, we only need to classify the specific content into exactly one EC among all caching candidates. Therefore, the sub-problem becomes a multi-classification problem. Next, we can define a CNN that corresponds to each sub-problem and that can be trained in parallel. 
In Section \ref{sec:subtraining}, we introduce the training process for each individual CNN and, then, in Section \ref{sec:subtesting}, we present the testing process by combining the prediction of all of the CNNs, and we propose two algorithms to improve performance.  

\subsection{Training Process}
\label{sec:subtraining}
As shown in Figure \ref{fig:training}, the training process includes four steps in general. Step ($i$) converts the original network caching allocation problem into an MILP model as discussed in Section \ref{sec:model}.

\begin{figure}[t]
	\centering
	\includegraphics[trim=0mm 0mm 0mm 5mm, clip, width=\textwidth]{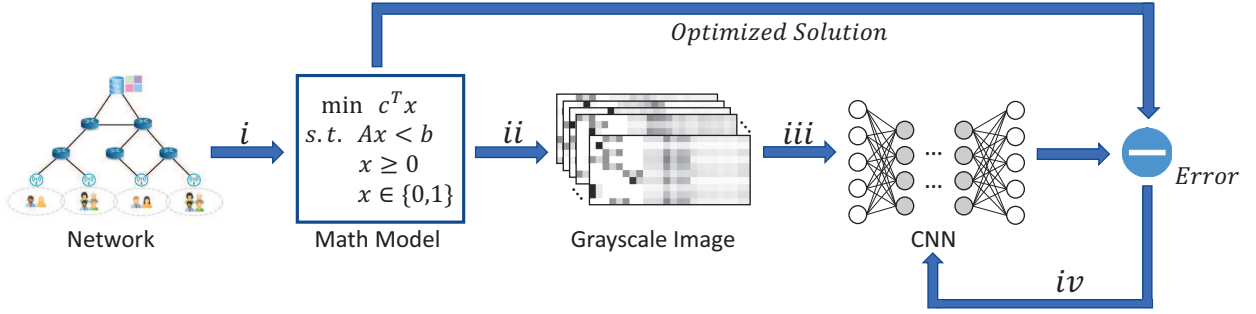}
	\caption{Training process for the proposed CNN.}
	\label{fig:training}
\end{figure}

\subsubsection{Grayscale Image Generation}
In step ($ii$), the mathematical model is transformed into a grayscale image which is used as an input to the CNN. An intuitive idea is constructing a dense parameter matrix which contains the entire information of the optimization model \eqref{LP:main_MILP}. This will require us to include the parameters as listed in Table \ref{tab:Notations} as part of the matrix. Therefore, our CNN should be wide and deep enough to recognize the large input image, which increases the complexity of CNN designing and training. We notice that some of the parameters in Table \ref{tab:Notations} are constant, such as $N_{ae}$ which is the number of hops that is typically fixed for a particular network topology. These parameters have very limited contribution to distinguish different assignments and can be viewed as redundant information during the training procedures.
By analysing the optimal solutions derived from \eqref{LP:main_MILP}, we observe that there is a pattern in user behaviour, network resources, and the allocation of content caching.
For instance, the EC is more likely to be selected as a caching host than others, if a) this EC is closer to mobile users, i.e. less transmission cost $C^T$, b) this EC has more available caching memory space, i.e. less caching cost $C^C$ and valid storage capacity constraint \eqref{LP:con2}, and c) this EC attaches links with enough bandwidth, i.e. valid bandwidth capacity \eqref{LP:con3}.
Furthermore, the transmission cost $C^T$ mentioned in a) depends on the moving probability $p_{ka}$ and the number of hops $N_{ae}$. 
For the storage space in b), we introduce $q_{ke}=s_k/w_e, \forall k\in\mathcal{K},e\in\mathcal{E}$ as a parameter that reflects the influence of caching content for flow $k$ on EC $e$. Therefore, constraint \eqref{LP:con2} can be transformed in the form of $q_{ke}$, which becomes $\sum_kq_{ke}\cdot x_{ke}<1, \forall e\in\mathcal{E}$. Clearly, larger available caching space $w_e$ results in smaller $q_{ke}$. Similarly, for the link bandwidth in c), we introduce a new parameter $r_{kl}=b_k/c_l,\forall k\in\mathcal{K},l\in\mathcal{L}$ to capture the link congestion level. We can now arrange $p_{ka}$, $q_{ke}$, and $r_{kl}$ as a matrix $\boldsymbol{T}$ that represents the features to be learned by CNN, as follows:
$$
\boldsymbol{T}=
\setlength{\arraycolsep}{1pt}
\left[
\begin{array}{ccc|ccc|ccc}
p_{11}&\cdots&p_{1|\mathcal{A}|}&q_{11}&\cdots&q_{1|\mathcal{E}|}&r_{11}&\cdots&r_{1|\mathcal{L}|}\\
p_{21}&\cdots&p_{2|\mathcal{A}|}&q_{21}&\cdots&q_{2|\mathcal{E}|}&r_{21}&\cdots&r_{2|\mathcal{L}|}\\
\vdots&\ddots&\vdots&\vdots&\ddots&\vdots&\vdots&\ddots&\vdots\\
p_{|\mathcal{K}|1}&\cdots&p_{|\mathcal{K}||\mathcal{A}|}&q_{|\mathcal{K}|1}&\cdots&q_{|\mathcal{K}||\mathcal{E}|}&r_{|\mathcal{K}|1}&\cdots&r_{|\mathcal{K}||\mathcal{L}|}\\
\end{array}
\right]
$$

Now, the element in the above matrix $\boldsymbol{T}$ can be viewed as the shades of gray, in which the range is from $0\%$ (white) to $100\%$ (black). Then, $\boldsymbol{T}$ is converted to a grayscale image.
In order to generate the matrix $\boldsymbol{T}$, we need to know the user moving probability $p_{ka}$. This probability can be predicted by exploring historical data, such as~\cite{chen2017caching} and \cite{al2018move}, and this technique is beyond the scope as well as complementary of this paper. Moreover, regarding the memory utilization $q_{ke}$ and link utilization $r_{kl}$, they can be achieved from the network management element. For instance, in the environment of multi-access edge computing (MEC) system, the virtualisation infrastructure manager in MEC host is responsible for managing storage and networking resources \cite{etsi2020mec}, i.e. $w_e$ and $c_l$. Once receiving the user required caching size $s_k$ and bandwidth $b_k$, we can calculate $q_{ke}$ and $r_{kl}$ accordingly.
Figure \ref{fig:encoding} shows a typical generated image where $|\mathcal{K}|=10$, $|\mathcal{A}|=8$, $|\mathcal{E}|=7$ and $|\mathcal{L}|=9$. In Figure \ref{fig:encoding}, the darker colour of a pixel, the larger is the value of that element in the matrix.
\begin{figure}[t]
	\centering
	\includegraphics[trim=3mm 12mm 8mm 0mm, clip, width=0.4\textwidth]{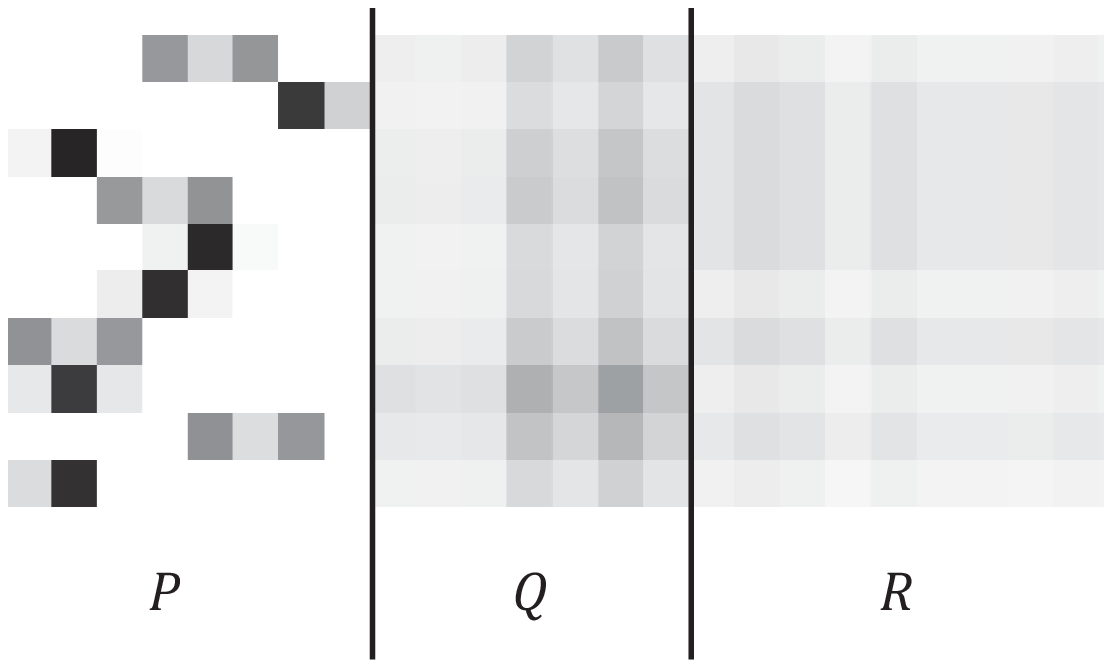}
	\caption{Feature encoding to grayscale image.}
	\label{fig:encoding}
\end{figure}
 
\subsubsection{Proposed CNN Structure}
\label{subsubsec:CNN}
The constructed image is applied as a feature to be learned by our CNN in step ($iii$). A common image recognition task has some key properties: a) some patterns are much smaller than the whole image; b) the same pattern appears in different positions; c) downsampling some pixels will not change the classification results. In particular, the convolutional layer satisfies properties a) and b) while the pooling layer matches c).
Our generated grayscale images also have properties a) and b). For example, the pixel in the middle of Figure \ref{fig:encoding} represents the effect on memory space utilization when caching content in the corresponding EC, which fits property a); and the dark pixel appears in different regions of the image, which matches property b).  
However, each pixel in our grayscale images provides essential information for content placement, and subsampling these pixels will change the CNN classification object. As a result, property c) does not hold in our images.
Therefore, we keep the convolutional layer in our CNN structure but we remove the pooling layer.
Figure \ref{fig:CNN} shows the structure of the proposed CNN which contains the following layers: 
\begin{figure}[t]
	\centering
	\includegraphics[trim=0mm 0mm 0mm 0mm, clip, width=\textwidth]{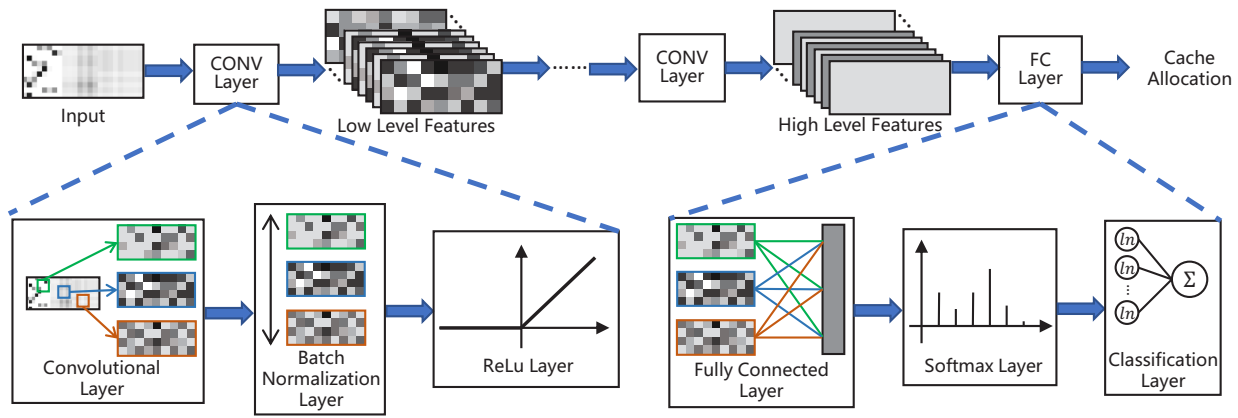}
	\caption{CNN architecture.}
	\label{fig:CNN}
\end{figure}
\begin{itemize}
\item \emph{Input Layer:} this layer specifies the image size and applies data normalization. In our model, the input is a grayscale image, so the channel size is one. The height is $|\mathcal{K}|$ and the width is $|\mathcal{A}|+|\mathcal{E}|+|\mathcal{L}|$.
\item \emph{CONV Layer:} this layer generates a successively higher-level abstraction of the input figure, or feature map (as commonly named). In our work, the CONV layer is composed of a) a convolutional layer, applying sliding convolutional filters to the input; b) a batch normalization layer, normalizing the input to speed up network training and reduce the sensitivity to network initialization; and c) a ReLu layer, introducing activation function rectified linear unit (ReLu) as the threshold to each element of the input. Generally, a CNN structure contains deep CONV layers. But we only use one CONV layer in our CNN. The effect of the different number of CONV layers is discussed in Section \ref{sec:investigations}.
\item \emph{FC Layer:} this layer produces the output of the CNN and consists of a) a fully connected layer that combines the features of a grayscale image to select the EC for caching; b) a softmax layer that normalizes the output of the preceding layer. It is worth noting that the output of softmax layer is a vector that contains non-negative numbers and the summary of vector elements equals to one. Therefore, the output vector can be used as a probability for classification; c) a classification layer that uses the output of the softmax layer for each grayscale image to assign it to one of the potential ECs and then compute the overall loss.
\end{itemize}

\subsubsection{Loss Function}
In step ($iv$), we employ the following cross-entropy loss function to estimate the gap between CNN prediction and the optimal solution, 
\begin{equation}
\label{eq:loss}
    \Upsilon(k)\!=\!-\!\sum_{i=1}^{|\mathcal{I}|}\sum_{e=1}^{|\mathcal{E}|}\!x^i_{ke}\!\ln{\hat{x}^i_{ke}}\!+\!\frac{\lambda}{2}\!\boldsymbol{W}^T\!\boldsymbol{W}, k\!=\!1,2,\cdots,|\mathcal{K}|,
\end{equation}
where $\mathcal{I}$ represents the training dataset, $\lambda$ is the L2 regularization (a.k.a Tikhonov regularization) factor to avoid overfitting, $\boldsymbol{W}$ is the weight vector, $x^i_{ke}$ is the optimal allocation of $i^{th}$ network scenario, and $\hat{x}^i_{ke}$ is the predicted result of CNN accordingly. 
During training, the weight vector $\boldsymbol{W}$ is updated recursively towards the reduction of loss function \eqref{eq:loss}.
Note that the original optimization problem \eqref{LP:main_MILP} contains five decision variables but we pick $x_{ke}$ as the label to train the CNN, because the rest can be approximated by $x_{ke}$.
For example, two auxiliary variables $\chi_{ke}$ and $t_e$ are determined via constraint \eqref{LP:con9}$\sim$\eqref{LP:con11} and \eqref{con:t_e} respectively. The upper-bound of $z_{kae}$ is limited by constraint \eqref{LP:con5}. Generally, the penalty of cache-miss is larger than the cache-hit transmission cost given that the data center is further away from AR than the host EC. Thus, most potential connected ARs would be taken into account except the AR with less attached probability (i.e. $p_{ka}$ is less than a threshold, which can be defined as the ratio between caching cost and caching gain~\cite{vasilakos2012proactive}). 
After that, the value of $z_{kae}$ is fixed. Moreover, constraints \eqref{LP:con6} and \eqref{LP:con7} provide the upper- and lower- bound of $y_{kl}$ in the form of $z_{kae}$.

\subsection{Testing Process}
\label{sec:subtesting}
Directly combining the output of each CNN may cause collisions since the aforementioned MILP model decomposition does not consider the internal effect of decision variables. In this subsection, two different approaches, reduced MILP (rMILP) and hill-climbing local search (HCLS), are proposed respectively to avoid the conflict.

\subsubsection{CNN-rMILP}
\label{sec:sub_ReducedMILP}
As seen from Figure \ref{fig:CNN-MILP}, the testing process consists of five steps where the first three steps are explained in Section \ref{sec:subtraining}.
\begin{figure}[t]
	\centering
	\includegraphics[trim=0mm 0mm 0mm 0mm, clip, width=.9\textwidth]{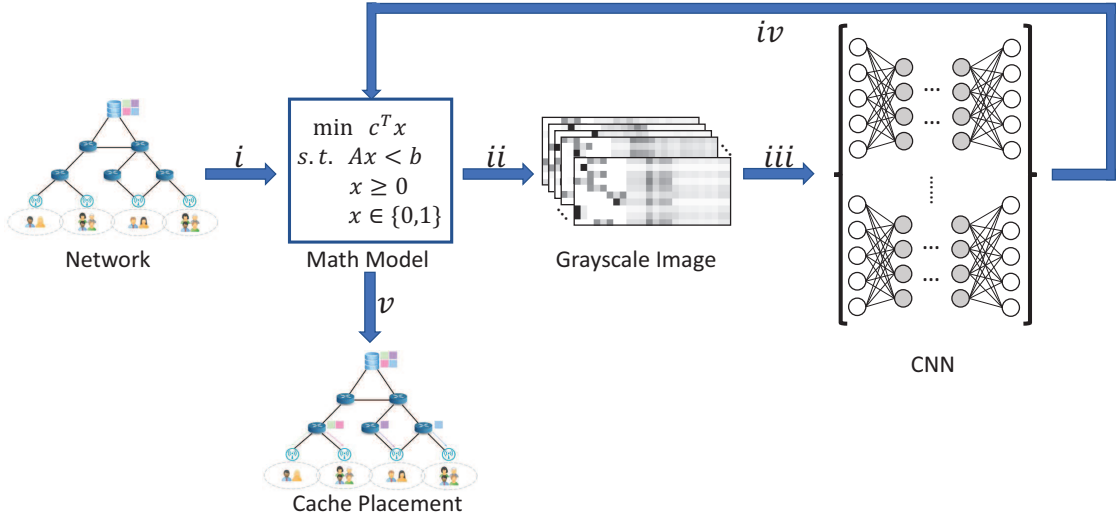}
	\caption{Testing process for CNN-rMILP.}
	\label{fig:CNN-MILP}
\end{figure}

In step ($iv$), the output of the softmax layer is a vector in which each element represents the probability of hosting content. This probability can be viewed as the confidence of prediction. For example, when the CNN output is $[0.8,0,0.15,0,0.05,0]$, then, the CNN is more confident to allocate the content in the first EC (first element in the vector) which has the highest value $0.8$. 
By combining these vectors, we can derive a matrix $\boldsymbol{O}=(o_{ke})_{|\mathcal{K}|\times|\mathcal{E}|}$. Here, each element $o_{ke}$ captures the probability of caching content for flow $k$ in EC $e$. 
We ignore the zero and iota elements because they are unlikely to be selected as caching host according to the CNN prediction. In order to filter out these elements, a unit step function $H(\cdot)$ is introduced with threshold $\delta$:
$$
H(o_{ke})=
\begin{cases}
1,  &o_{ke}\leq\delta,  \\
0,  &o_{ke}>\delta. 
\end{cases}
$$

After the filter $H(o_{ke})$, all $1$ elements are viewed as candidates to host content. As done in~\cite{lei2019learning}, we perform a global search by adding a new constraint in the original optimization problem \eqref{LP:main_MILP}, which becomes:
\begin{subequations} \label{LP:new}
	\begin{align}
		&\mathop{\min} J(\chi_{ke},z_{kae}). \\
		\textrm{s.t.}\quad &\eqref{LP:con1}\sim\eqref{LP:con8},\eqref{LP:con9}\sim\eqref{LP:con11},\eqref{con:t_e},\eqref{LP:con12}, \nonumber\\
		\label{LP:added}
		&x_{ke}\leq H(o_{ke}),\forall k\!\in\!\mathcal{K},e\!\in\!\mathcal{E}.	
	\end{align}
\end{subequations}

In step ($v$), the new optimization problem \eqref{LP:new} can be efficiently solved when $\big(H(o_{ke})\big)_{|\mathcal{K}|\times|\mathcal{E}|}$ is a sparse matrix since the feasible region of the decision variable $x_{ke}$ is greatly reduced. Note that, although the introduction of constraint \eqref{LP:added} enables a reduction in search space, the feasible solutions may also be eliminated by constraint \eqref{LP:added}, especially when the CNN is not well-trained or the threshold $\delta$ in function $H(\cdot)$ is high. As a result, problem \eqref{LP:new} becomes unsolvable. 
In case of infeasibility, constraint \eqref{LP:added} is slacked and we need to resolve the original optimization model \eqref{LP:main_MILP}. 

\subsubsection{CNN-HCLS}

\begin{figure}[t]
\centering
\includegraphics[trim=0mm 25mm 0mm 0mm, clip, width=.9\textwidth]{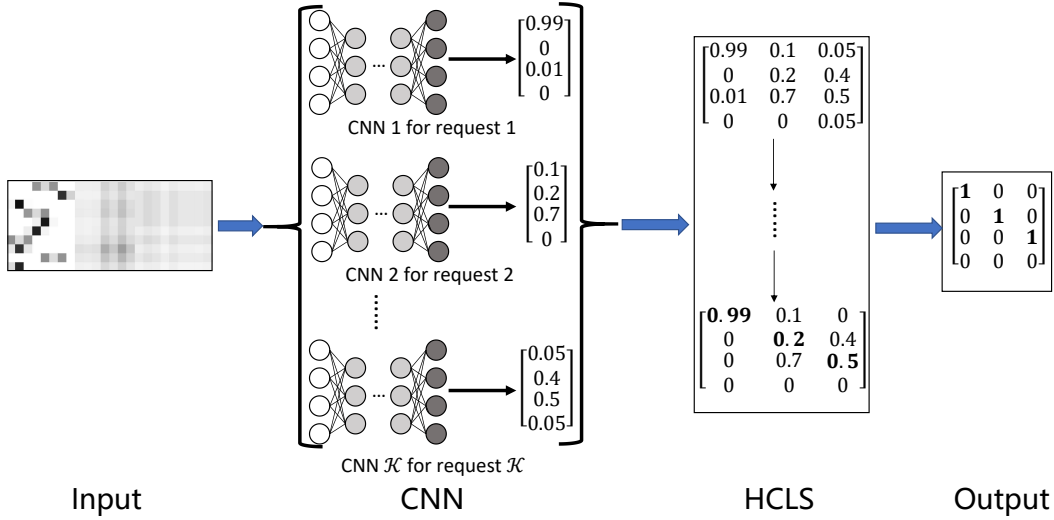}
\caption{Testing process for CNN-HCLS.}
\label{fig:CNN-HCLS}
\end{figure}

The CNN-HCLS process is shown in Figure \ref{fig:CNN-HCLS}.
We define the \textit{state space} as a collection of all possible solutions and non-solutions, i.e. all assignment possibilities of variable $x_{ke}$; the \textit{successor state} is the ``locally'' accessible states from a current state with the largest probability, where ``locally'' means the successor state only differs by one-user cache allocation; finally, the \textit{score} $S$ is defined as a piecewise linear function which is added to the total cost $J$ as a penalty cost to measure the impact of invalid constraints \eqref{LP:con2} and \eqref{LP:con3} in the form of $q_{ke}$ and $r_{kl}$:
\begin{equation}
\label{fml:new_cost}
S\!=\!\gamma\!\max\!\left\{\!0\!,\!\left(\!\sum_{e\in\mathcal{E}}\!\sum_{k\in\mathcal{K}}\!q_{ke}\!x_{ke}\!-\!1\!\right)\!,\!\left(\!\sum_{l\in\mathcal{L}}\!\sum_{k\in\mathcal{K}}\!r_{kl}\!y_{kl}\!-\!1\!\right)\!,\!\left(\!\sum_{e\in\mathcal{E}}\!\sum_{k\in\mathcal{K}}\!q_{ke}\!x_{ke}\!-\!1\!\right)\!+\!\left(\!\sum_{l\in\mathcal{L}}\!\sum_{k\in\mathcal{K}}\!r_{kl}\!y_{kl}\!-\!1\!\right)\!\right\}\!+\!J\!,
\end{equation}
where $\gamma$ is the penalty factor for invalidated constraints. Similar to the processing in CNN-rMILP, to reduce the number of attempted assignments, the matrix is filtered via the unit step function $H(\cdot)$ with threshold $\delta$. The details are described in Algorithm 1.

{\renewcommand\baselinestretch{1.35}
\begin{algorithm}[t]
\caption{Hill-Climbing Local Search (HCLS)}
\label{alg:HCLS}
\KwData{predicted conditional probability $\boldsymbol{O}=(o_{ke})$, threshold of probability $\delta$, variables in Table \ref{tab:Notations}}
\KwResult{flow assignment $\boldsymbol{X}=(x_{ke})$}
Construct assignment matrix $\boldsymbol{X}\leftarrow \boldsymbol{X}\left(o_{ke}\geq\delta\right)$\;\label{alg:state1}
Select the largest number on each column as the initial state\;\label{alg:state2}
Evaluate the cost $S$ via \eqref{fml:new_cost} of the successor states\;\label{alg:state3}
Move to a successor state with less cost\;\label{alg:state4}
Repeat from Step \ref{alg:state3} until no further improvement of cost are possible\;\label{alg:state5}
Set all the elements in the terminated state to be 1, while the rest 0\; \label{alg:state6}
\end{algorithm}
\par}

Considering Figure \ref{fig:HCLS-instance} as an instance, we combine the output of CNNs in Figure \ref{fig:CNN-HCLS} as a $4\times3$ matrix, where rows indicate candidate ECs and the columns represent different requests. In step \ref{alg:state1} of Algorithm \ref{alg:HCLS}, each predicted probability is compared with threshold $\delta=0.1$ and all element less than this value are set to zero. Then, we select the largest value (shown in red and bold font) in each column as the initial case, i.e. $0.99$, $0.7$ and $0.5$ in three columns respectively. 
Since $0.99$ is the first element in the first column, the initial assignment will allocate the first request in the first EC. Similarly, the second and third requests are cached in the third EC. 
In step \ref{alg:state2} because $0.99$ is the only non-zero item in the first column, there is no successor in this direction. For the second column, the second-largest number is $0.2$ with position index (2,2) in the initial state, so one successor state is keeping the allocation of first and third requests in the original location but caching the second request in the second EC, i.e. branch $i$ in Figure \ref{fig:HCLS-instance}. Similarly, we get another successor state following the third column as branch $ii$. In the following two steps, we evaluate the cost $S$ of these three assignments via formula \eqref{fml:new_cost}, where successor state $i$ results in a smaller cost. Then, we check further successor states such as $iii$ and $iv$ recursively based on current state $i$. In this example, the algorithm terminates at the local optimal solution (shown in a blue dot rectangular) because further improved states cannot be found. In the last step, the bold and red colored positions are selected as caching ECs.

\begin{figure}[t]
\centering
\includegraphics[trim=0mm 160mm 0mm 130mm, clip, width=\textwidth]{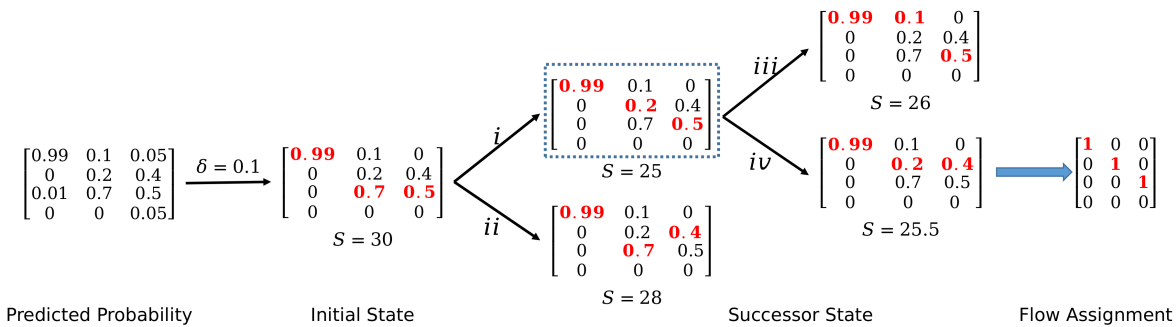}
\caption{An instance of CNN-HCLS.}
\label{fig:HCLS-instance}
\end{figure}

As can be seen from the above instance, there are at most $|\mathcal{K}|$ branches from the previous state to the successor state, and the maximum number of states is $|\mathcal{E}|$, which represents all potential ECs are enumerated. Moreover, the computation of cost $S$ in each branch takes constant time. Therefore, the time complexity of Algorithm \ref{alg:HCLS} is $O(|\mathcal{K}|\cdot|\mathcal{E}|)$. Once the network topology is determined, $|\mathcal{E}|$ becomes constant and, then, the running time is reduced to $O(|\mathcal{K}|)$, i.e. linear time complexity.

In summary, we have proposed two different approaches, CNN-rMILP and CNN-HCLS, to decide cache content placement. While both of them utilize the output of trained CNN, the CNN plays different roles: in CNN-rMILP, CNN is used to eliminate improbable ECs as constraint \eqref{LP:added}; in CNN-HCLS, CNN provides the searching direction in Algorithm \ref{alg:HCLS}. Note that the proposed CNN-rMILP and CNN-HCLS provide sub-optimal solutions of \eqref{LP:main_MILP}, and should not be considered as solving a NP-hard problem with the optimal solution.
\section{Numerical Results}
\label{sec:investigations}

\subsection{Simulation Setting}
We will now compare the performance of the optimal decision making derived from the MILP model \eqref{LP:main_MILP} with the proposed CNN-rMILP and CNN-HCLS. To further benchmark the results we also compare with a greedy caching algorithm (GCA) which attempts to allocate each request to its nearest EC, as detailed in Algorithm \ref{alg:greedy}. Additionally, to address the improvements of CNN-rMILP and CNN-HCLS schemes, we report the resulting quality of using only the CNN (named as pure-CNN in the following). We assume a nominal mesh tree-like mobile network topology where user mobility is taking place on the edge of the network. Based on the mobility patterns of the users, we apply the different techniques, i.e., MILP (benchmark), pure-CNN, CNN-rMILP, CNN-HCLS, and GCA respectively, to perform proactive edge cloud caching of popular content. 
The simulation parameters are summarized in Table \ref{tab:Network_Parameters}. 

{\renewcommand\baselinestretch{1.35}
\begin{algorithm}[t]
\caption{Greedy Caching Algorithm (GCA)}
\label{alg:greedy}
\KwData{variables in Table \ref{tab:Notations}}
\KwResult{flow assignment $\boldsymbol{X}=(x_{ke})$}
\For{each $k\in\mathcal{K}$}{
Find the maximum $p_{ka}$ and related $a$, $p_{ka}\!\leftarrow\!0$\;
Build an EC queue $L_e$ from $a$ by Dijkstra algorithm\;
	\For{each $e\in L_e$}{
		\If{$s_k \leq w_e$}{
			$x_{ke}\!\leftarrow\!1$, $w_e\!\leftarrow\!w_e\!-\!s_k$\;
			\textbf{break}\;
		}
	}
}
\end{algorithm}
\par}

\begin{table}[t]
\centering
\caption{\label{tab:Network_Parameters}Network parameters \cite{wang2019proactive}.}
\begin{tabular}{l|c}
\hline
\textbf{Parameter}&\textbf{Value} \\
\hline
Degree per node & 1$\sim$5 \\
Number of mobile users ($|\mathcal{K}|$)& \{5,10,15,20\} \\
Number of links ($|\mathcal{L}|$)& 20\\
Number of access routers ($|\mathcal{A}|$)& 7\\
Number of edge clouds ($|\mathcal{E}|$)& 6\\
Threshold of prediction probability ($\delta$) & 0.001 \\
Size of user request content ($s_k$) & [10,50] MB \\
Available cache size in EC ($w_e$) & [100,500] MB\\
User request transmission bandwidth ($b_k$) & [1,10] Mbps\\
Link available capacity ($c_l$) & [50,100] Mbps\\
\hline
\end{tabular}
\end{table}

The dataset is generated by solving the optimization problem \eqref{LP:main_MILP} via nominal branch-and-bound approaches with different user behaviour and network utilization levels, i.e., $p_{ka}$, $s_k$, $b_k$, $w_e$ and $c_l$, which follow uniform distributions. In the subsequent results, we first generate 1280 samples for a scenario with $5$ users, out of which $1024$ samples are used for training by using the structure of each CNN illustrated in Figure \ref{fig:CNN}; $128$ samples are used as the validation set for hyperparameter tuning, as illustrated in Section \ref{subsec:hyper}; and the remaining $128$ samples are used for performance testing, which is discussed in Section \ref{subsec:compa}. 
For the scenarios with $10$, $15$, and $20$ users, we construct $128$ samples for each case accordingly.
We note that the input image size of the trained CNN is $5\times33$ (i.e. $|\mathcal{K}|\times(|\mathcal{A}|+|\mathcal{E}|+|\mathcal{L}|)$) and the images for more than $5$ requests exceed this size.  
For the purpose of matching the input layer of the trained CNN, the large grayscale image can be divided into partitions and the height of each sub-image is considered as 5 to fit the input size.
Next, the trained CNNs (i.e. the so-called pure-CNN) are applied for predicting the allocation of sub-images and, then, updating the unassigned sub-images based on the assignment in a cascade manner. 
The behaviour of the traffic flow such as the moving probability $p_{ka}$ is independent from the caching locations, and, hence, it will remain constant during the update process. For the update of $q_{ka}$, we have   
\begin{equation}
\label{eq:update_ec}
    q_{ke}=\frac{s_k}{w_e-\sum_{k'\in\mathcal{K'}}s_{k'}x_{k'e}}, \forall k\in\mathcal{K}\setminus\mathcal{K'},
\end{equation}
where $\mathcal{K'}$ is the set of traffic flows who have already been assigned. Regarding the link utilization $r_{kl}$, the update is calculated in a similar fashion,
\begin{equation}
\label{eq:update_link}
    r_{kl}=\frac{b_k}{c_l-\sum_{k'\in\mathcal{K'}}b_{k'}y_{k'l}}, \forall k\in\mathcal{K}\setminus\mathcal{K'}.
\end{equation}

From \eqref{eq:update_ec}, we can see that the value of $q_{ke}$ could become negative which indicates that the previous allocation is invalid. This, in turn, results in a congestion at EC $e$, because it violates the constraint $w_e\!-\!\sum_{k'\in\mathcal{K'}}s_{k'}x_{k'e}\!>\!0$. The same holds also for variable $r_{kl}$ as well. In this special case, the EC or link is not involved in next sub-image assignment. 
We keep caching contents and updating network parameters $q_{ka}$ as well as $r_{kl}$ until all users' requests are satisfied, or we terminate the cache placement when no additional network resources (i.e. caching memory space, link bandwidth) are available.
The pure-CNN, instead of CNN-HCLS or CNN-rMILP, is used to update the image in order to avoid error accumulation. This can be explained from Figure \ref{fig:HCLS-instance}, where the output of the pure-CNN configuration is a probability matrix and many different assignments (as long as the cell is not zero) could be taken into consideration but for CNN-HCLS and CNN-rMILP, the end result is a binary matrix. Once the allocation for the previous sub-image is inappropriate, the error is accumulated and it affects the remaining sub-images. The operations when the number of end-users exceeds 5 users are summarized in Algorithm \ref{alg:update}.

{\renewcommand\baselinestretch{1.35}
\begin{algorithm}[t]
\caption{Augmenting Allocations for $K$ Requests}
\label{alg:update}
\KwData{Grayscale image $\boldsymbol{I}$}
\KwResult{Flow assignment $\boldsymbol{X}=(x_{ke})$}
Separate the image into $\ceil*{|\mathcal{K}|/5}$ sub-images\;
\While{any sub-image not assigned}{
Call pure-CNN to do prediction for one of the unassigned sub-image\;
Update unassigned sub-images based on \eqref{eq:update_ec} and \eqref{eq:update_link}\;
}
Combine the predictions into a $|\mathcal{K}|\times|\mathcal{E}|$ matrix\;
Call HCLS (Algorithm \ref{alg:HCLS}) or rMILP (Section \ref{sec:sub_ReducedMILP})\;
\label{alg:STATE_call}
\end{algorithm}
\par}

In contrast to the nominal multi-label classification problem, in our case, some misclassification is still acceptable if the total cost $S$ can be deemed as competitive. 
We first compare the computation complexity among these algorithms which is measured by the average testing time. For pure-CNN, the time tracking is from the grayscale image loading in the input layer to the prediction produced in output layer excluding the training time. When it comes to CNN-rMILP and CNN-HCLS, the processing time also includes the global search by solving reduced MILP and local search through hill-climbing respectively. Moreover, the average total cost with invalid constraint penalties \eqref{fml:new_cost} is calculated. Then, we derive the feasible ratio which is defined as the percentage of assignments satisfying the constraints. In turn, this will shed light on the network congestion after allocation. For example, $100\%$ indicates adequate network availability, while $60\%$ indicates a significant congestion. 

Additionally, the maximum total cost difference with the optimal solution is analyzed, which represents the performance gap in the worst case. We also compare the average number of decision variables for the benchmark and the CNN-rMILP in the $128$ testing samples, which reflects the search space reduction. Finally, some typical classification evaluation metrics are included, such as accuracy, precision, recall, and $F_1$ score\footnote{To put computational times in perspective we note that simulations run on MATLAB 2019b in a 64-bit Windows 10 environment on a machine equipped with an Intel Core i7-7700 CPU 3.60 GHz Processor and 16 GB RAM}. According to \cite{zhang2013review}, four basic quantities, namely the true positive ($T^+$), false positive ($F^+$), true negative ($T^-$) and false negative ($F^-$) are defined as follows, for each EC $e$:
\begin{align*}
&T^+_e\!=\!\sum_{i=1}^{|\mathcal{T}|}|\pi_{ie}\!\in\!X_i\!\cap\!\pi_{ie}\!\in\!\hat{X}_i|,\quad
F^+_e\!=\!\sum_{i=1}^{|\mathcal{T}|}|\pi_{ie}\!\notin\!X_i\!\cap\!\pi_{ie}\!\in\!\hat{X}_i|, \\
&T^-_e\!=\!\sum_{i=1}^{|\mathcal{T}|}|\pi_{ie}\!\notin\!X_i\!\cap\!\pi_{ie}\!\notin\!\hat{X}_i|,\quad
F^-_e\!=\!\sum_{i=1}^{|\mathcal{T}|}|\pi_{ie}\!\in\!X_i\!\cap\!\pi_{ie}\!\notin\!\hat{X}_i|,
\end{align*}
where $\mathcal{T}$ represents the testing data set, $|\mathcal{T}|$ is the number of testing samples, $\pi_{ie}$ is the caching allocation, $X_i$ is the optimal solution of $i^{th}$ sample and $\hat{X}_i$ express the output of estimated algorithms. Based on these quantities, related metrics, such as accuracy ($A$), precision ($P$), recall ($R$) and $F_1$ score, can be calculated with
\begin{align*}
&A(T^+_e,F^+_e,T^-_e,F^-_e)=\frac{T^+_e+T^-_e}{T^+_e+F^+_e+T^-_e+F^-_e}, \quad
P(T^+_e,F^+_e,T^-_e,F^-_e)=\frac{T^+_e}{T^+_e+F^+_e},\\ &R(T^+_e,F^+_e,T^-_e,F^-_e)=\frac{T^+_e}{T^+_e+F^-_e}, \quad
F_1(T^+_e,F^+_e,T^-_e,F^-_e)=\frac{2\times T^+_e}{2\times T^+_e+F^-_e+F^+_e}.
\end{align*}

Let $f$ be one of the four metrics, i.e., $A$, $P$, $R$ or $F_1$. The macro-averaged and micro-averaged version of $f$ are calculated as follows:
\begin{equation*}
\begin{aligned}
f_{\textrm{macro}}=\frac{1}{|\mathcal{E}|}\sum_{e=1}^{|\mathcal{E}|}f(T^+_e,F^+_e,T^-_e,F^-_e), \quad f_{\textrm{micro}}=f(\sum_{e=1}^{|\mathcal{E}|}T^+_e,\sum_{e=1}^{|\mathcal{E}|}F^+_e,\sum_{e=1}^{|\mathcal{E}|}T^-_e,\sum_{e=1}^{|\mathcal{E}|}F^-_e).
\end{aligned}
\end{equation*}
The macro and micro accuracy $A$ are equivalent depending on the definition with the fact that $T\!P_e+T\!N_e+F\!P_e+F\!N_e=|\mathcal{T}|$. 
For the rest micro metrics, we have the following theorem: 
\begin{theorem}
	\label{theo:micro}
	The precision, recall, and $F_1$ score in micro-averaged version are equivalent.
\end{theorem}
\begin{proof}
	See Appendix \ref{sec:proof_B}.
\end{proof}
Hence, in the following evaluations, the micro accuracy is omitted and the micro recall together with $F_1$ score are represented by the micro precision. 

\subsection{Hyperparameters Tuning}
\label{subsec:hyper}
The parameters in the CNN can be divided into two categories: normal parameters, such as the neuron weights, which can be learned from the training dataset; and hyperparameters, like the depth of the CONV layers, whose value should be determined during the CNN architecture design to control the training process. These parameters affect the CNN's output and, as a result, the overall network performance. Therefore, the hyperparameters have to be tuned.
After several rounds of grid search and by considering the tradeoff between simplicity and loss function, the best configuration recorded is the following: CONV layer = $1$, batch size = $64$, epoch = $30$, learning rate = $10^{-3}$. In the sequel, the impact and sensitivity of each hyperparameter on the trained CNN is investigated, whilst all other hyperparameters are fixed and the metric is set to be the loss function as previously defined in \eqref{eq:loss}.

\begin{figure}[htbp]
    \centering
    \subfigure[Training loss function with different depth.]{
        \label{fig:depth_train}
        \includegraphics[width=\figSize\textwidth]{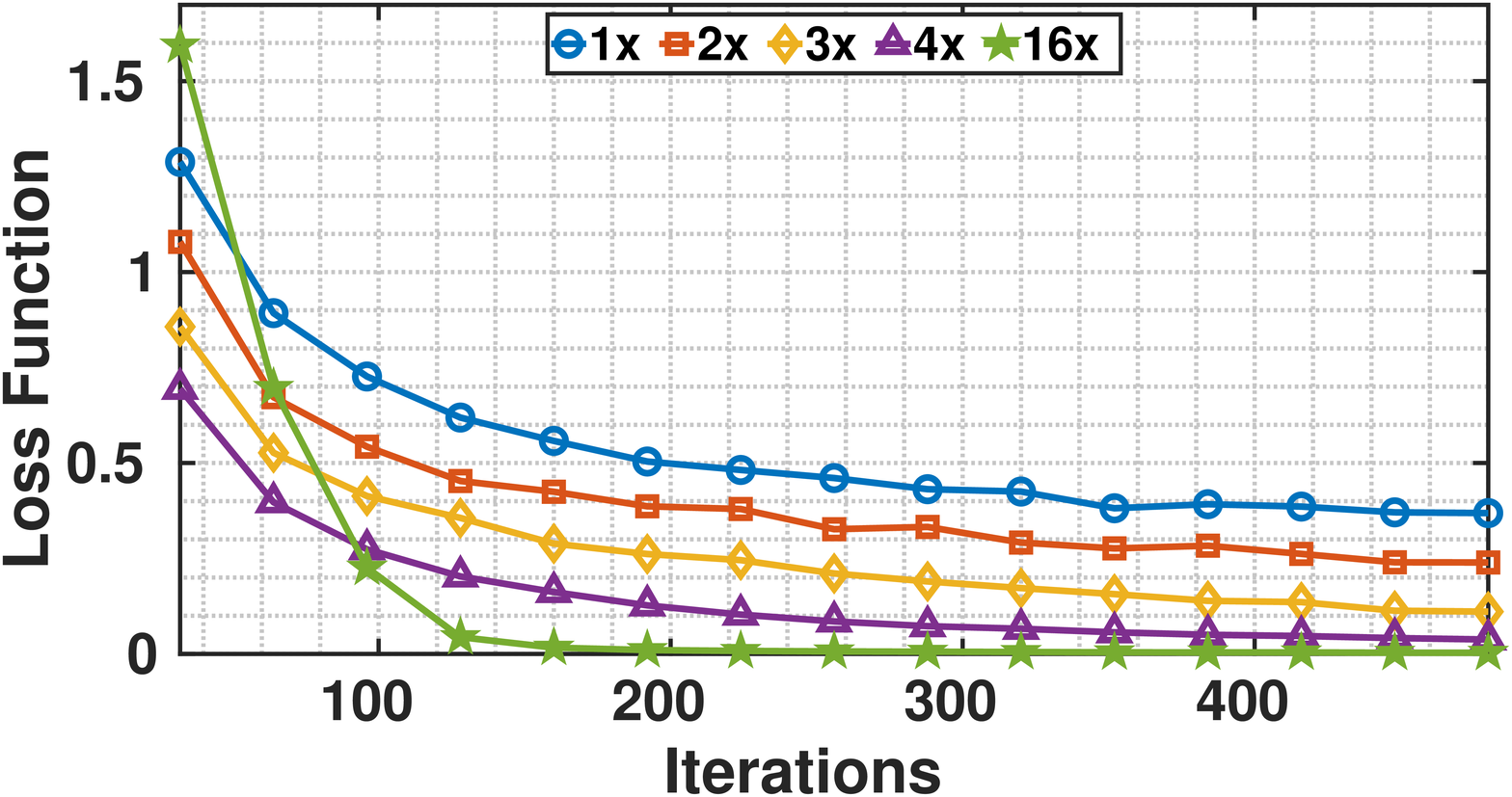}
    }
    \subfigure[Validation loss function with different depth.]{
	    \label{fig:depth_valid}
        \includegraphics[width=\figSize\textwidth]{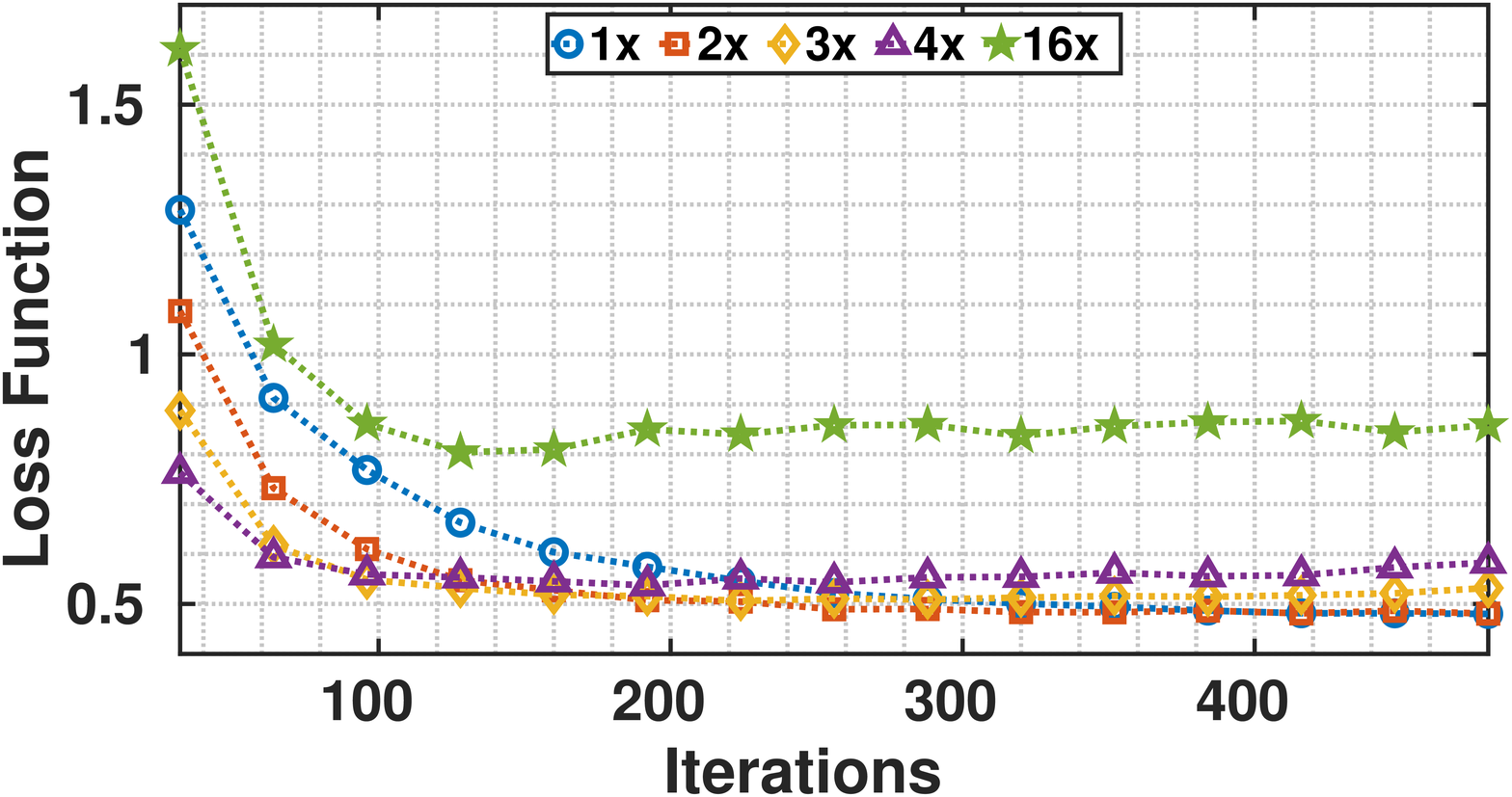}
    }
    \quad  
    \subfigure[Training time with different depth (Iterations=$480$).]{
    	\label{fig:depth_time}
        \includegraphics[width=\figSize\textwidth]{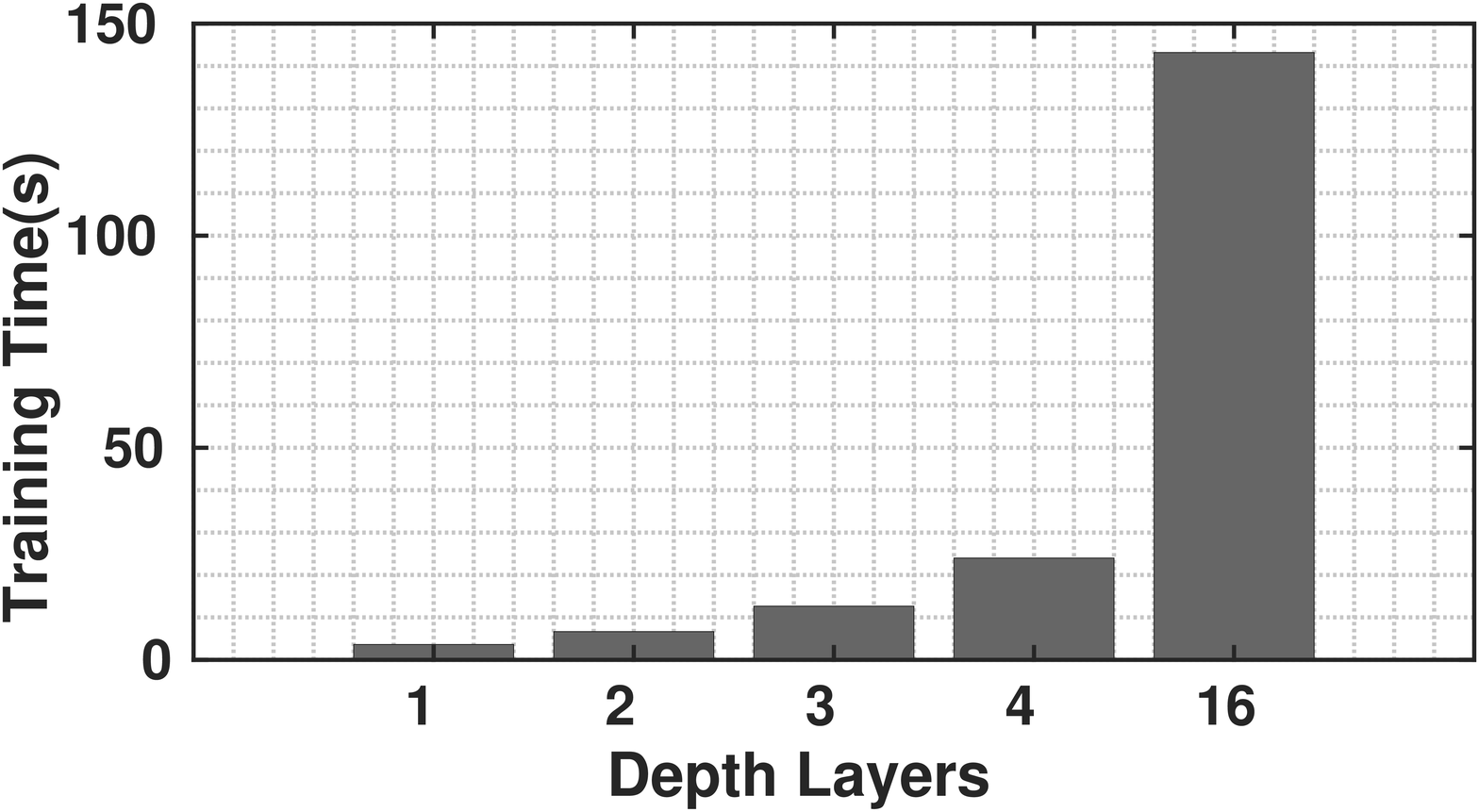}
    }
    \subfigure[Training time with different batch size (Epoch=$30$).]{
	    \label{fig:batch_time}
        \includegraphics[width=\figSize\textwidth]{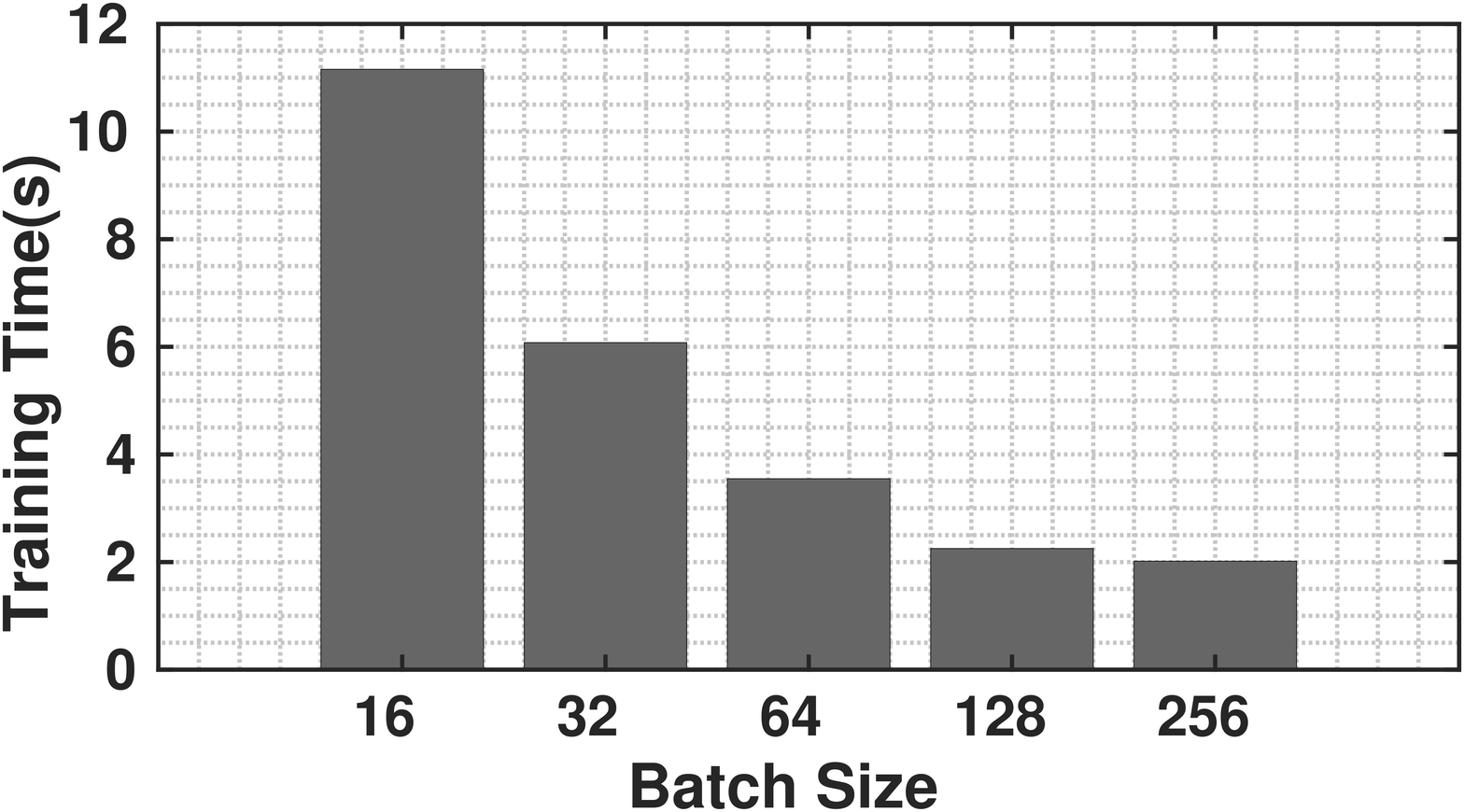}
    }
    \quad 
    \subfigure[Training loss function with different batch size.]{
    	\label{fig:batch_train}
        \includegraphics[width=\figSize\textwidth]{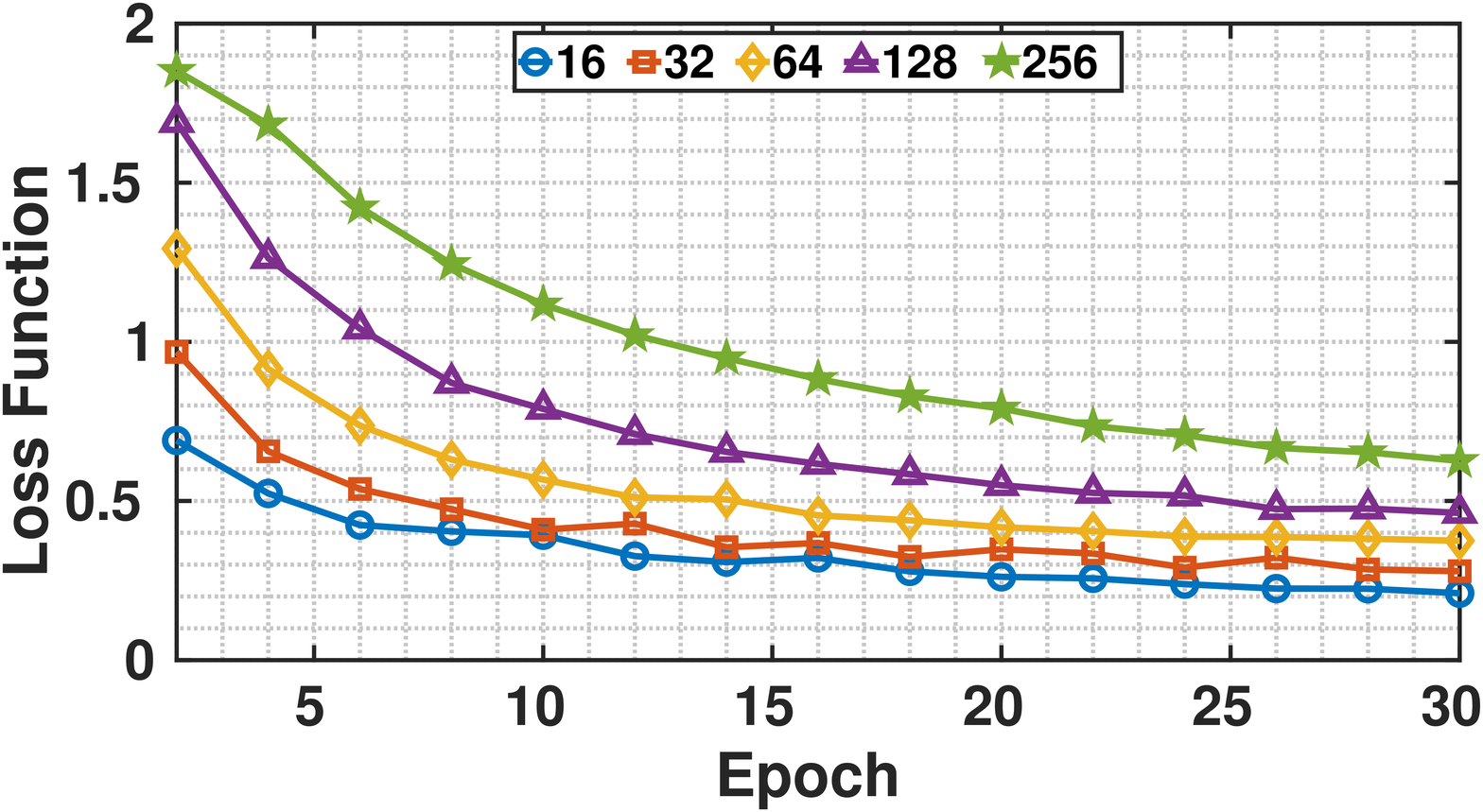}
    }
    \subfigure[Validation loss function with different batch size.]{
	    \label{fig:batch_valid}
        \includegraphics[width=\figSize\textwidth]{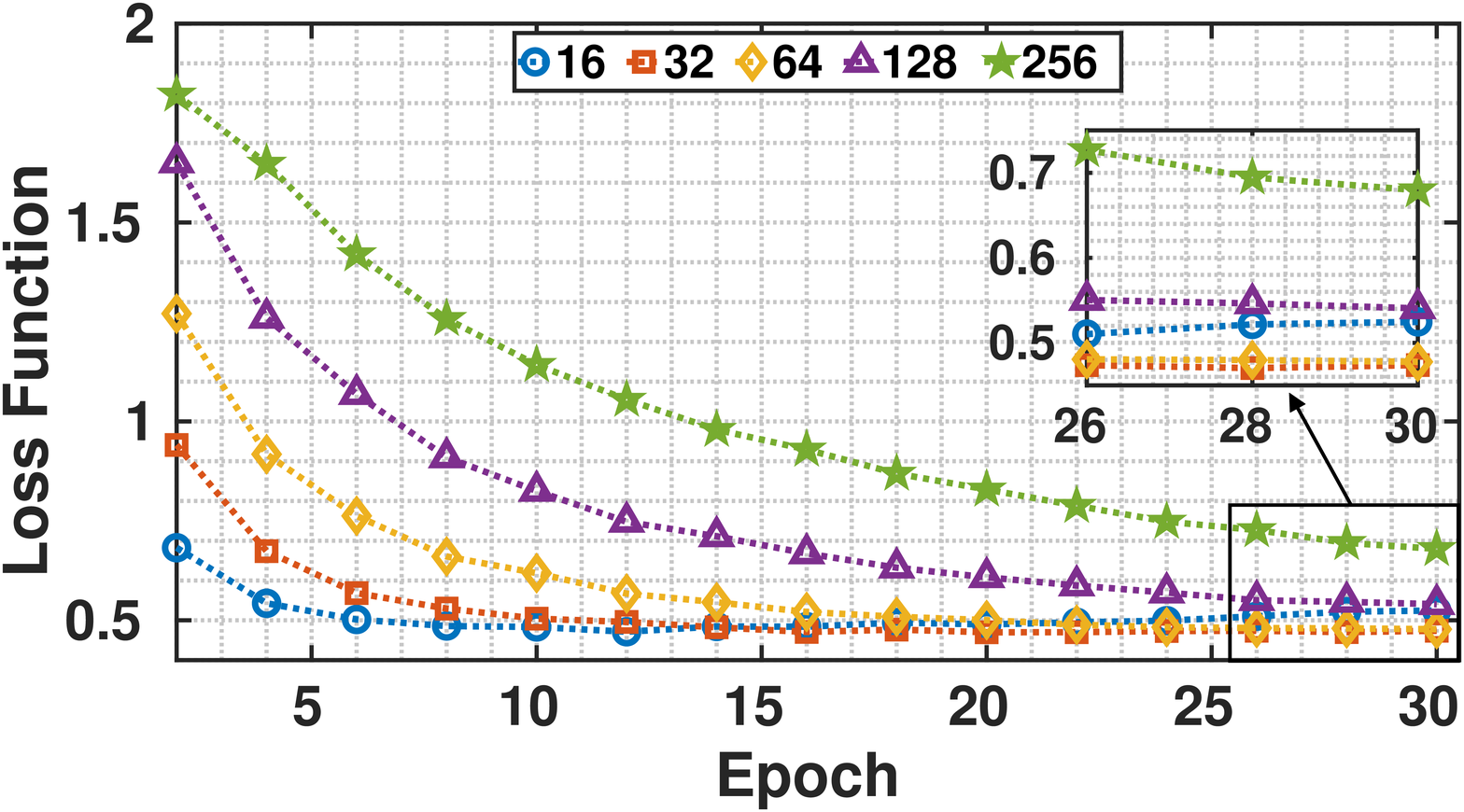}
    }
    \quad   
    \subfigure[Loss function with different epochs.]{
    	\label{fig:epoch}
        \includegraphics[width=\figSize\textwidth]{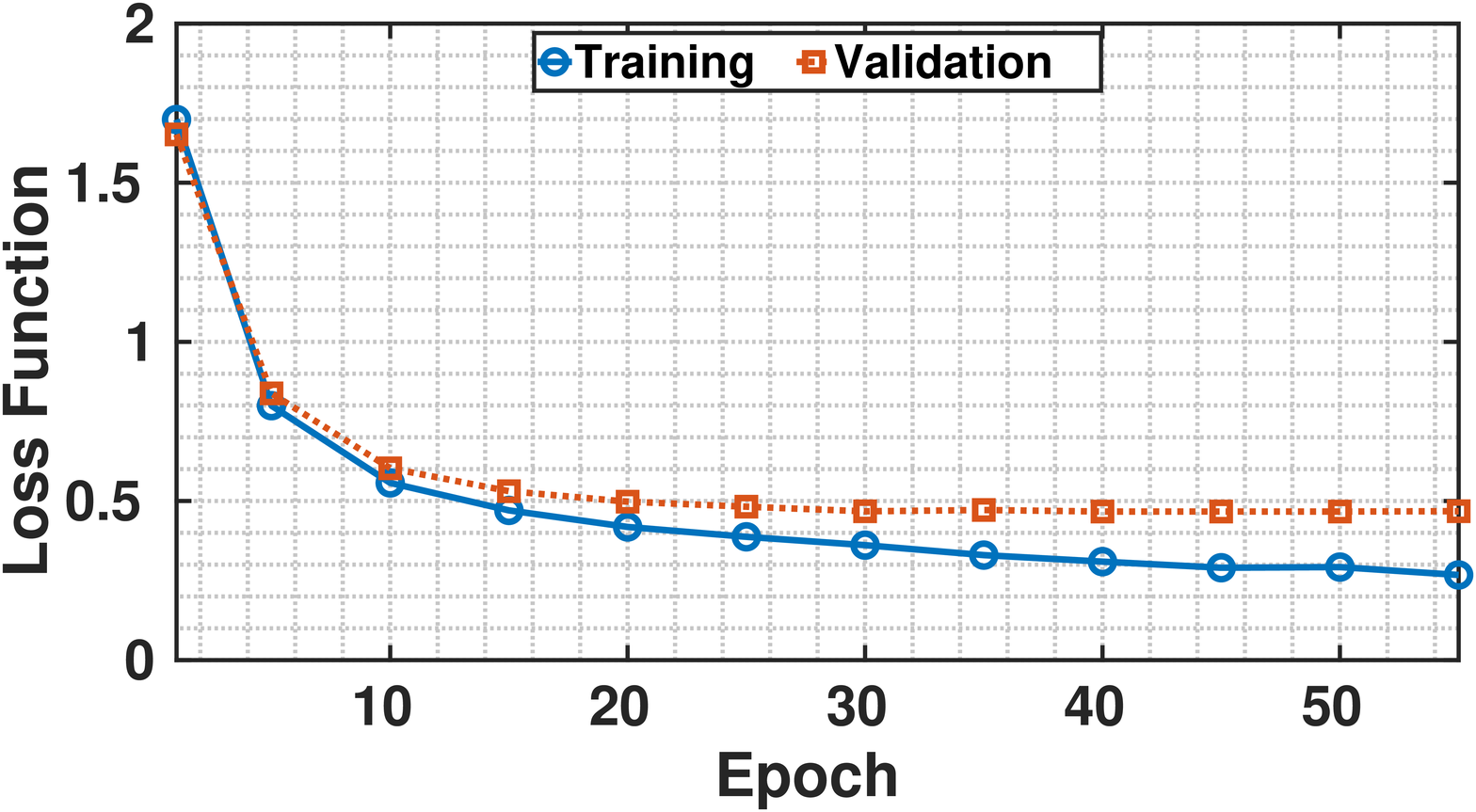}
    }
    \subfigure[Loss function with different learning rate.]{
	    \label{fig:learning_rate}
        \includegraphics[width=\figSize\textwidth]{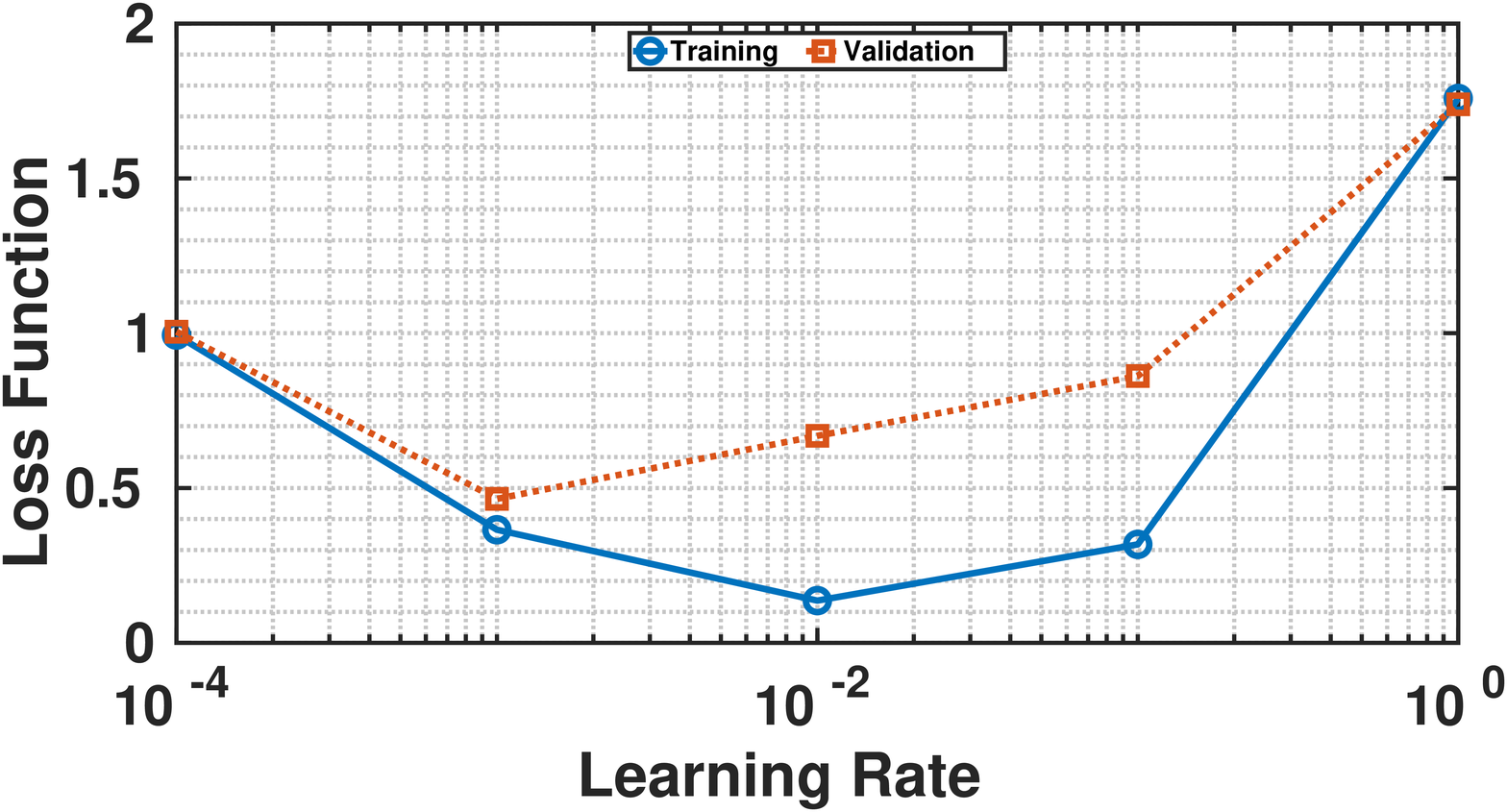}
    }
    \caption{The effect of hyperparameters.}
    \label{fig:hyper}
\end{figure}

\subsubsection{Effect of CONV Depths}
As already eluded to above, each CONV layer encompasses a convolutional layer, a normalization layer, and a ReLU layer. The effect of CONV layers on training and performance are shown in Figures \ref{fig:depth_train}, \ref{fig:depth_valid}, and \ref{fig:depth_time} respectively. More complex structures do not bring obvious advantages over the 1 CONV layer structure but add computational complexity (the more CONV layers, the more training time required in Figure \ref{fig:depth_time}). Additionally, the loss function deteriorates when adding more layers. For instance, the curve of 16 hidden layers in the validation figure \ref{fig:depth_valid} is above the rest. 
In light of the above findings, the hidden layer depth in the sequel is set to 1.

\subsubsection{Effect of the Batch Size}
The batch size is by definition the number of training data in each iteration. In the fixed epochs and training dataset, larger  batch sizes  lead to smaller iterations, which reduces the training time as Figure \ref{fig:batch_time} shows. According to \cite{keskar2016large}, larger batch sizes tend to converge to sharp minimizers but smaller batch sizes would terminate at flat minimizers and the later has better generalization abilities. In Figures \ref{fig:batch_train} and \ref{fig:batch_valid}, we compare the performance for different batch sizes, namely sizes of 16, 32, 64, 128 and 256. As expected, the small batch size has better performance in general. Specifically, around the 30-epoch instance in Figure \ref{fig:batch_valid}, the 32-batch curve outperforms the 64-batch (less than $8.2$\textperthousand) while the others achieve a worse performance. However, note that the training time for 32-batch is almost double that of the 64-batch as shown in Figure \ref{fig:batch_time}. To this end, for the sake of performance and to enable less running time in terms of epochs, the batch size is chosen to be 64.

\subsubsection{Effect of Epochs}
Figure \ref{fig:epoch} illustrates the loss function in training and validation. Initially, the loss drops heavily by increasing the epoch number. Then, the generalization capability decreases slowly from 10 to 30 epochs and reach a plateau after 30 epochs. In this case, large epochs can improve the performance slightly (less than $1.5$\textperthousand) but the price to pay is temporal cost, i.e., more than $86.1\%$ training time is required. Therefore, we select 30 epochs for the training. Note that overfitting could be eased by introducing $L2$ type regularization in \eqref{eq:loss}. 

\subsubsection{Effect of the Learning Rate}
Figure \ref{fig:learning_rate} shows the effect of the learning rate. As expected, a lower learning rate can explore the downward slope with a small step but it would take a longer time to converge. When the learning rate is $10^{-4}$, the training process ends before convergence because the epoch is set to 30; and this can be seen as the reason why a lower learning rate is under-performing. On the other hand, a larger learning rate may oscillate over the minimum point, leading to a potential failure to converge and even diverge, such as the case of the learning rate reaching $1$ in Figure \ref{fig:learning_rate}. Consequently, the learning rate is set to be $10^{-3}$ because of the lowest validation loss.

\subsection{Caching Performance Comparison}
\label{subsec:compa}
We first compare the performance of the different schemes for the case of 5 requests, as illustrated in Table \ref{tab:5req}. The optimal performance of each testing metric is shown in a bold font. In general, all evaluated methods can achieve solutions efficiently, even the branch-and-bound technique (benchmark) provides optimal solutions in around 100 ms. Therefore, we use the optimal solutions in the case of 5 requests to establish training dataset instead of the case of 10, 15 or 20 requests. Benefiting from the prior knowledge learned during the training process and the simple structure (only 1 CONV layer), the pure-CNN configuration can provide the decision making in just 2 ms. However, the pure-CNN approach deals with each flow request independently and does not consider the cross correlation among these requests. Compared with benchmark, while the pure-CNN saves approximately $98.4\%$ computation time, the network suffers from more than $300\%$ total cost payment and $1.4\%$ constraints lack of satisfaction.

The proposed CNN-rMILP utilizes the prediction from pure-CNN to reduce the search space. As can be seen from Table \ref{tab:5req}, the number of decision variables in CNN-rMILP is reduced to $356.98$ on average, while the original number in the  benchmark is $376$. Compared with the benchmark, this reduction saves almost $71.6\%$ computation time, with just $0.8\%$ additional total cost.  Even in the worst case, the performance gap between CNN-rMILP and the optimal solution is $0.32$, which is much smaller than any other algorithm. Therefore, CNN-rMILP provides highly competitive performance for the case of 5 requests. Similarly, CNN-HCLS use the output of the pure-CNN to guide the feasible solution searching, which leads to around $94.6\%$ running time reduction but $149.6\%$ cost increment.

To provide a more holistic view, we tested the above schemes in terms of different nominal metrics such as accuracy, precision, recall and $F_1$ score in Table \ref{tab:5req}.
Interestingly, although the pure-CNN approach performs better in the macro- and micro-averaged metrics, the greedy algorithm GCA and the proposed CNN-HCLS can provide better solutions with less total costs. 
This trend also holds for cases of 10, 15, and 20 requests.
Thus, there is no strong correlation between total cost and typical deep learning metrics, such as accuracy, precision, recall, and $F_1$. 

\begin{table*}[!t]
\centering
\caption{Performance comparison for the case of 5 requests.}
\label{tab:5req}
\begin{tabular}{c|c|c|c|c|c|c}
\hline
\multicolumn{1}{r}{} & &\textbf{Benchmark} &\textbf{pure-CNN} &\textbf{CNN-rMILP} &\textbf{CNN-HCLS} &\textbf{GCA} \\
\hline
\multicolumn{2}{c|}{Mean time} & 101.7 ms & \textbf{1.6 ms} & 28.9 ms & 5.5 ms & 9.6 ms \\
\hline
\multicolumn{2}{c|}{Mean total cost} & \textbf{7.41} & 31.11 & 7.47 & 18.50 & 23.76 \\
\hline
\multicolumn{2}{c|}{Mean feasible ratio} & \textbf{100.00\%} & 98.56\% & \textbf{100.00\%} & 98.83\% & 98.80\% \\
\hline
\multicolumn{2}{c|}{Max total cost difference} & \textbf{0} & 243.51 & 0.32 & 120.90 & 303.20 \\
\hline
\multicolumn{2}{c|}{Number of decision variables} & 376.00 & - & \textbf{356.98} & - & - \\
\hline
\multicolumn{2}{c|}{Macro/Micro accuracy} & \textbf{100.00\%} & 99.52\% & 99.76\% & 93.90\% & 91.35\% \\
\hline
\multicolumn{1}{c|}{\multirow{3}{*}{Macro}}  & precision & \textbf{100.00\%} & 65.87\% & 97.76\% & 58.09\% & 55.28\% \\
\cline{2-7} & recall & \textbf{100.00\%} & 65.87\% & 90.78\% & 60.41\% & 60.73\% \\
\cline{2-7} & $F_1$ & \textbf{100.00\%} & 65.86\% & 93.09\% & 56.74\% & 57.24\% \\
\hline
\multicolumn{1}{c|}{Micro} & precision/recall/$F_1$ & \textbf{100.00\%} & 98.56\% & 99.27\% & 81.71\% & 74.06\% \\
\hline
\end{tabular}
\end{table*}

\begin{table*}[!t]
\centering
\caption{Performance comparison for the case of 10 requests.}
\label{tab:10req}
\begin{tabular}{c|c|c|c|c|c|c}
\hline
\multicolumn{1}{r}{} & &\textbf{Benchmark} &\textbf{pure-CNN} &\textbf{CNN-rMILP} &\textbf{CNN-HCLS} &\textbf{GCA} \\
\hline
\multicolumn{2}{c|}{Mean time} & 2842.8 ms & \textbf{3.6 ms} & 987.7 ms & 19.1 ms & 16.5 ms \\
\hline
\multicolumn{2}{c|}{Mean total cost} & \textbf{22.71} & 269.65 & 23.91 & 120.55 & 166.33 \\
\hline
\multicolumn{2}{c|}{Mean feasible ratio} & \textbf{100.00\%} & 83.86\% & \textbf{100.00\%} & 86.52\% & 86.12\% \\
\hline
\multicolumn{2}{c|}{Max total cost difference} & \textbf{0} & 703.36 & 48.51 & 239.23 & 936.56 \\
\hline
\multicolumn{2}{c|}{Number of decision variables} & 746.00 & - & \textbf{647.74} & - & - \\
\hline
\multicolumn{2}{c|}{Macro/Micro accuracy} & \textbf{100.00\%} & 92.90\% & 95.22\% & 88.64\% & 89.97\% \\
\hline
\multicolumn{1}{c|}{\multirow{3}{*}{Macro}} & precision & \textbf{100.00\%} & 57.01\% & 73.62\% & 55.01\% & 55.59\% \\
\cline{2-7} & recall & \textbf{100.00\%} & 56.70\% & 73.57\% & 55.71\% & 56.28\% \\
\cline{2-7} & $F_1$ & \textbf{100.00\%} & 55.59\% & 73.48\% & 54.04\% & 55.43\% \\
\hline
\multicolumn{1}{c|}{Micro} & precision/recall/$F_1$ & \textbf{100.00\%} & 78.69\% & 85.67\% & 65.91\% & 69.92\% \\
\hline
\end{tabular}
\end{table*}

\begin{table*}[!t]
\centering
\caption{Performance comparison for the case of 15 requests.}
\label{tab:15req}
\begin{tabular}{c|c|c|c|c|c|c}
\hline
\multicolumn{1}{r}{} & &\textbf{Benchmark} &\textbf{pure-CNN} &\textbf{CNN-rMILP} &\textbf{CNN-HCLS} &\textbf{GCA} \\
\hline
\multicolumn{2}{c|}{Mean time} & 35095.0 ms & \textbf{6.0 ms} & 7636.4 ms & 49.3 ms & 23.1 ms \\
\hline
\multicolumn{2}{c|}{Mean total cost} & \textbf{62.83} & 667.05 & 67.24 & 271.45 & 366.11 \\
\hline
\multicolumn{2}{c|}{Mean feasible ratio} & \textbf{100.00\%} & 67.80\% & \textbf{100.00\%} & 74.01\% & 72.96\% \\
\hline
\multicolumn{2}{c|}{Max total cost difference} & \textbf{0} & 1404.95 & 128.50 & 471.16 & 1049.46 \\
\hline
\multicolumn{2}{c|}{Number of decision variables} & 1116.00 & - & \textbf{883.36} & - & - \\
\hline
\multicolumn{2}{c|}{Macro/Micro accuracy} & \textbf{100.00\%} & 86.16\% & 88.60\% & 83.49\% & 87.10\% \\
\hline
\multicolumn{1}{c|}{\multirow{3}{*}{Macro}} & precision & \textbf{100.00\%} & 43.92\% & 59.37\% & 46.26\% & 53.21\% \\
\cline{2-7} & recall & \textbf{100.00\%} & 46.26\% & 59.56\% & 44.33\% & 53.19\% \\
\cline{2-7} & $F_1$ & \textbf{100.00\%} & 43.68\% & 59.32\% & 44.71\% & 53.02\% \\
\hline
\multicolumn{1}{c|}{Micro} & precision/recall/$F_1$ & \textbf{100.00\%} & 58.47\% & 65.79\% & 50.48\% & 61.30\% \\
\hline
\end{tabular}
\end{table*}

\begin{table*}[!t]
\centering
\caption{Performance comparison for the case of 20 requests.}
\label{tab:20req}
\begin{tabular}{c|c|c|c|c|c|c}
\hline
\multicolumn{1}{r}{} & &\textbf{Benchmark} &\textbf{pure-CNN} &\textbf{CNN-rMILP} &\textbf{CNN-HCLS} &\textbf{GCA} \\
\hline
\multicolumn{2}{c|}{Mean time} & 41.04 min & \textbf{10.1 ms} & 4.91 min & 103.7 ms & 29.7 ms \\
\hline
\multicolumn{2}{c|}{Mean total cost} & \textbf{111.17} & 1160.65 & 127.48 & 461.10 & 567.68 \\
\hline
\multicolumn{2}{c|}{Mean feasible ratio} & \textbf{100.00\%} & 55.68\% & \textbf{100.00\%} & 64.51\% & 64.00\% \\
\hline
\multicolumn{2}{c|}{Max total cost difference} & \textbf{0} & 2023.86 & 527.45 & 946.13 & 1268.07 \\
\hline
\multicolumn{2}{c|}{Number of decision variables} & 1486.00 & - & \textbf{1110.88} & - & - \\
\hline
\multicolumn{2}{c|}{Macro/Micro accuracy} & \textbf{100.00\%} & 82.63\% & 83.76\% & 80.74\% & 85.48\% \\
\hline
\multicolumn{1}{c|}{\multirow{3}{*}{Macro}} & precision & \textbf{100.00\%} & 37.68\% & 49.33\% & 40.71\% & 53.48\% \\
\cline{2-7} & recall & \textbf{100.00\%} & 42.18\% & 48.88\% & 39.60\% & 52.92\% \\
\cline{2-7} & $F_1$ & \textbf{100.00\%} & 37.44\% & 48.93\% & 39.92\% & 53.13\% \\
\hline
\multicolumn{1}{c|}{Micro} & precision/recall/$F_1$ & \textbf{100.00\%} & 47.89\% & 51.29\% & 42.23\% & 56.45\% \\
\hline
\end{tabular}
\end{table*}

Tables \ref{tab:10req}, \ref{tab:15req}, and \ref{tab:20req} show different performance metrics for the case of 10, 15, and 20 requests respectively. Generally, the performance of all estimated algorithms deteriorates as the number of requests increases. Although the benchmark provides the minimum cost in all cases, the computation time increases rapidly. As shown in Table \ref{tab:10req}, the computational time for the optimal decision making is less than $3$ seconds in the case of 10 requests, and, then, the time reaches around $41.04$ minutes on average for the case of 20 requests in Table \ref{tab:20req}. Regarding the pure-CNN scheme, on the one hand, it is the fastest method among the five estimated algorithms, whose average decision time is just $10.1$ ms even in the case of 20 requests; on the other hand, it performs worse than GCA, where the total cost of pure-CNN is double the GCA's cost. Moreover, nearly half of the flow requests cannot be served by the allocation from pure-CNN, as the feasible ratio is $55.68\%$ in Table \ref{tab:20req}.

The performance of the proposed CNN-rMILP and CNN-HCLS schemes also become progressively worse due to the increment of user requests. With regard to CNN-rMILP, the mean total cost gap with the benchmark increases from $5.0\%$ (the case of 10 requests) to $14.7\%$ (the case of 20 requests), and the computation time increases from around $1$ second to $5$ minutes. 
For CNN-HCLS, under the case of 20 requests, it can provide decision making in $103.7$ ms with the expense of additional $300\%$ total cost payment.
In summary, compared with the benchmark, both CNN-rMILP and CNN-HCLS make a tradeoff between processing time and overall cost. In particular, CNN-rMILP is suitable for cases in which pseudo-real time processing is allowed so as to provide high-quality decision making, whereas the CNN-HCLS is amenable to real-time decision making with performance gap tolerance.  

\section{Conclusions}
\label{sec:conclusions}

In this paper, we have proposed a framework in which a data-driven technique in the form of deep convolutional neural networks (CNNs) is amalgamated with a model-based approach in the form of a mixed integer programming problem. We have transformed the optimization model into a grayscale image and, then, trained CNNs in parallel with optimal decisions. To further improve the performance, we have provided two algorithms: the proposed CNN-HCLS is suitable for the case in which the processing time is more significant; while CNN-rMILP fits the scenario where the overall cost is more important. 
Numerical investigations reveal that the proposed schemes provide competitive cache allocations. In the case of 15 flow requests, CNN-rMILP manages to accelerate the calculation by 5 times compared to the benchmark, with less than $7\%$ additional overall cost. 
Furthermore, the average computational time for CNN-HCLS is only $49.3$ msec, which make the proposed framework suitable for real-time decision making. 

We envision several future directions for this work. First, we have assumed that the collected network information is convincing. When there are controlled small variations in the input, we need further sensitivity analysis of the deep learning output. Furthermore, considering a scenario with time sequential content request, an interesting extension of the proposed CNN would be to include temporal characteristics, which will require new ways to transform the spatio-temporal characterization of the caching placement problem into an image. Finally, one can study the design of novel distributed machine learning algorithms \cite{9210812,9457160} so as to enable edge devices to proactively optimize caching content and delivery methods without transmitting a large amount of data.

\appendix

\subsection{Proof of Theorem \ref{theo:np_hard}}
\label{sec:proof_A}
	Suppose we have $\mathcal{G_s}\!=\!(\mathcal{V_s} ,\mathcal{L_s})$, where $\mathcal{V_s}\!=\!\{e_1,e_2,a_1\}$ and $\mathcal{L_s}\!=\!\{(e_1,a_1),(e_2,a_1)\}$. For the purpose of expressing easily, we define $l_1\!=\!(e_1,a_1)$ and $l_2\!=\!(e_2,a_1)$. We also define $\mathcal{E_s}\!=\!\{e_1,e_2\}$, $\mathcal{A_s}\!=\!\{a_1\}$, and $\mathcal{L_s}\!=\!\{l_1,l_2\}$.
	Assume the available storage capacities for EC is infinity, i.e., $w_e\to+\infty$. Therefore, constraint \eqref{LP:con2} can be relaxed. Noting that 
	$$\lim_{w_e\to+\infty}\frac{s_k}{w_e}=0,$$
	the constraint \eqref{con:t_e} can be removed safely as well because  $t_e\equiv1$ in that case. Moreover, constraints \eqref{LP:con9}$\sim$\eqref{LP:con11} can be relaxed considering $\chi_{ke}=t_e\!\cdot\!x_{ke}=x_{ke}$. In the proposed network topology $\mathcal{G_s}$, there is only one AR $a_1$ in the network, so the second dimension in decision variable $z_{kae}$ could be squeezed as $z_{ke}$. As a result,  the decision variables $x_{ke}$, $y_{kl}$, and $z_{kae}$ have the exact same meaning in this scenario, since each path consists only one link (i.e. the path from $a_1$ to $e_1$ is link $l_1$, and the path from $a_1$ to $e_2$ is $l_2$) and each link represents a unique path. In other words, we can use decision variable $y_{kl}$ to replace the other two. Furthermore, let the moving probability $p_{ka}=1,\forall k\in\mathcal{K}$. Then the problem \eqref{LP:main_MILP} can be expressed as: given a positive number $b$, is it possible to find a feasible solution $y_{kl}$ such that
	
	\begin{equation}
	\label{fml:obj2}
	J\leq b,
	\end{equation}
	with constraints
	\begin{subequations}
		\begin{align}
		\label{fml:newcon1}
		&\sum_{l\in\mathcal{L_s}}y_{kl}\!\leq\!1,\forall k\!\in\!\mathcal{K}, \\
		\label{fml:newcon2}
		&\sum_{k\in\mathcal{K}}b_k\!\cdot\!y_{kl}\!<\!c_l,\forall l\!\in\!\mathcal{L_s}, \\
		\label{fml:newcon3}
		&y_{kl}\!\in\!\{0,1\},\forall k\!\in\!\mathcal{K},l\!\in\!\mathcal{L_s},
		\end{align}
	\end{subequations}
	satisfied. Same as \cite{fan2018application}, assume $b\to+\infty$, then \eqref{fml:obj2} can be relaxed. Moreover, we set each link has same capacity, i.e. $c_{l_1}=c_{l_2}=\frac{1}{2}\sum_{k\in\mathcal{K}}b_k+\nu$, where $\nu\ll\frac{1}{2}min(b_k)$. In order to satisfy \eqref{fml:newcon2}, we guarantee that $$\sum_{k\in\mathcal{K}}b_k y_{kl_1}=\sum_{k\in\mathcal{K}}b_k y_{kl_2}=\frac{1}{2}\sum_{k\in\mathcal{K}}b_k,$$
	i.e. the set-partition problem.	In other word, it is reducible to MILP model \eqref{LP:main_MILP} while set-partition is a well-known $NP$-hard problem \cite{cormen2013algorithms}.

\subsection{Proof of Theorem \ref{theo:micro}}
\label{sec:proof_B}
	In the operation of micro, all the quantities, i.e. $T^+$, $T^-$, $F^+$ and $F^-$, are in the form of summary alongside EC. Regarding $F^+$, there is
	\begin{equation*}
	\begin{aligned}
	F^+&=\sum_{e=1}^{|\mathcal{E}|}F^+_e=\sum_{e=1}^{|\mathcal{E}|}\sum_{i=1}^{|\mathcal{T}|}|\pi_{ie}\notin X_i\cap\pi_{ie}\in\hat{X}_i|=\sum_{i=1}^{|\mathcal{T}|}\left(\sum_{e=1}^{|\mathcal{E}|}|\pi_{ie}\notin X_i\cap\pi_{ie}\in\hat{X}_i|\right).
	\end{aligned}
	\end{equation*}
	The cardinality $|\pi_{ie}\notin X_i\cap\pi_{ie}\in\hat{X}_i|$ indicates the number of cells whose value is $1$ in prediction but $0$ in fact in each testing instance. Constraint \eqref{LP:con1} forces that each request is served by exact one EC. In other words, once the cell is $1$ in prediction but $0$ in fact, there would be a corresponding cell in the same column is $0$ but $1$ in fact. Therefore,
	$$
	\sum_{e=1}^{|\mathcal{E}|}|\pi_{ie}\notin X_i\cap\pi_{ie}\in\hat{X}_i|=\sum_{e=1}^{|\mathcal{E}|}|\pi_{ie}\in X_i\cap\pi_{ie}\notin\hat{X}_i|, \forall i\in\mathcal{T}.
	$$
	$F^+$ can be rewritten as
	\begin{align*}
	   F^+=&
	   \sum_{i=1}^{|\mathcal{T}|}\left(\sum_{e=1}^{|\mathcal{E}|}|\pi_{ie}\notin X_i\cap\pi_{ie}\in\hat{X}_i|\right)=
	   \sum_{i=1}^{|\mathcal{T}|}\left(\sum_{e=1}^{|\mathcal{E}|}|\pi_{ie}\in X_i\cap\pi_{ie}\notin\hat{X}_i|\right)\\ 
	   =&\sum_{e=1}^{|\mathcal{E}|}\sum_{i=1}^{|\mathcal{T}|}|\pi_{ie}\in X_i\cap\pi_{ie}\notin\hat{X}_i|=\sum_{e=1}^{|\mathcal{E}|}F^-_e=F^-.
	\end{align*}
	Consequently,
	\begin{equation*}
	\begin{aligned}
	    P_{\textrm{micro}}
	    =\frac{\sum_{e=1}^{|\mathcal{E}|}T^+_e}{\sum_{e=1}^{|\mathcal{E}|}T^+_e+\sum_{e=1}^{|\mathcal{E}|}F^+_e}
	    =\frac{T^+}{T^++F^+}=\frac{T^+}{T^++F^-}=R_{\textrm{micro}},
	\end{aligned}
	\end{equation*}
	and 
	\begin{equation*}
	\begin{aligned}
		F_{1\textrm{micro}}\!=\!\frac{2\!\times\!T^+}{2\!\times\!T^+\!+\!F^-\!+\!F^+}
 		=\frac{2\times T^+}{2\times T^++F^++F^+}
		\!=\!\frac{2\!\times\! T^+}{2\!\times\!(T^+\!+\!F^+)}=P_{\textrm{micro}}.
	\end{aligned}
	\end{equation*}

\bibliographystyle{ieeetr}
\def\baselinestretch{1.4}
\bibliography{reference}
\end{document}